\documentclass[11pt]{article}

\usepackage[a4paper,margin=1in]{geometry}
\usepackage{graphicx,amsmath}
\usepackage{authblk}
\usepackage[linesnumbered,ruled,vlined,algo2e]{algorithm2e}
\usepackage{amssymb}
\usepackage{tabularx}
\usepackage{multirow}
\usepackage{makecell}
\usepackage{array}
\usepackage{wrapfig}
\usepackage{hyperref}
\usepackage{amsthm}

\newtheorem{theorem}{Theorem}
\newtheorem{lemma}[theorem]{Lemma}

\newtheorem{proposition}[theorem]{Proposition}

\newtheorem{claim}[theorem]{Claim}

\theoremstyle{definition}
\newtheorem{definition}[theorem]{Definition}

\theoremstyle{remark}
\newtheorem{remark}[theorem]{Remark}

\newcommand{\reject}{\texttt{reject}}
\newcommand{\accept}{\texttt{accept}}
\newcommand{\framework}{\textbf{\texttt{refix}}}

\title{Distributed Local Verification using Proofs with(out) Errors}

\author{
Pawe\l\ Garncarek$^{1}$\thanks{\texttt{pgarn@cs.uni.wroc.pl}},
Tomasz Jurdzi\'nski$^{1}$\thanks{\texttt{tju@cs.uni.wroc.pl}},
Dariusz Kowalski$^{2}$\thanks{\texttt{darek.liv@gmail.com}},
Subhajit Pramanick$^{1}$\thanks{\texttt{suvo.iitg17@gmail.com}}
}

\date{
$^{1}$Institute of Computer Science, University of Wroc\l aw, Poland\\
$^{2}$Department of Computer Science, Augusta University, USA
}

\begin{document}

\maketitle

\begin{abstract}
  
  We investigate well-known \textit{local verification} problems concerning graph properties in distributed networks under the framework of \textit{locally checkable proofs} (LCP): if a graph has a given property, then there exists a proof, assigned by a \emph{prover} on the graph nodes, such that the verifier (a distributed algorithm running at every node) makes all nodes accept, whereas if the property is not satisfied, then for any proof on the graph, the verifier makes at least one node reject.
  Nodes generally access a constant ($\geq 1$) size local neighborhood, referred to as \emph{view distance}, to make a decision.

  Our focus is twofold. First, among multiple properties we consider in the paper, we put particular emphasis on verifying cycle existence -- whether the graph has a cycle or not (not the same as checking cycle-freeness) -- mainly because this can be verified with only $3$ labels and access to $1$ hop neighborhood, with a matching lower bound. 
  What is more interesting is that, inspired by the widely used technique of encoding a sense of direction (induced by BFS distances from a selected root) along a path used in many verification problems, we introduce a novel gadget (by repeated use of a special string $001101$) for encoding direction using only $2$ labels and a view distance of $3$ for each node. Although we use it for cycle existence, this gadget may be of independent interest for other verification tasks.

  Second, we consider an  \emph{erroneous proof} model, where an adversary can corrupt proof labels (assigned by a trusted prover or \emph{oracle}) of at most $i$ nodes within the $2i+1$-hop neighborhood of each node.
    We present an algorithmic framework, called \textbf{\texttt{\textbf{refix}}}, that transforms the verification algorithm for the error-free variant of the LCP to tolerate the aforementioned error model, at the cost of $2i+1$ view distance per node.
    We illustrate the applicability of the framework to three graph properties: cycle existence, cycle-freeness, and bipartiteness. 
    We also provide lower bounds to establish the relation between the number of errors and the view distance for every node.
    Finally, as a first step toward studying LCPs in the CONGEST model, we show that cycle-existence verification with $2$ proof labels and view distance $3$ can be implemented using only a $3$-round algorithm, suggesting a promising direction for extending LCPs beyond the LOCAL model.
\end{abstract}


\newpage

\tableofcontents

\newpage

\setcounter{page}{1}

\section{Introduction}
\label{sec:introduction}

In many distributed systems, nodes are required not only to compute a solution, but also to \emph{locally verify} that a desired global structure or property is present in the network. Such tasks give rise to \emph{verification problems}, where the goal is to determine whether a graph satisfies a specified property using only local information, where the term  \emph{graph property} corresponds to a family of graphs closed under isomorphism. In the context of yes--no verification problems, a standard and well-accepted model \cite{goos2016locally} requires the following: (i) for a yes-instance, all nodes must output $1$; and (ii) for a no-instance, at least one node must output $0$.

In distributed graph algorithms, the model in which nodes are provided with task-specific auxiliary information is well established and widely studied~\cite{10.1145/3382734.3405715,balliu_et_al:LIPIcs.DISC.2021.8,10.1145/3188745.3188860,bousquet_et_al:LIPIcs.STACS.2024.21,10.1145/1073814.1073817,feuilloley2022error}. A popular direction in this template is the notion of \emph{proof labeling schemes} (PLS), introduced by Korman et al. \cite{10.1145/1073814.1073817}, which is a verification scheme involving a \emph{prover} and a \emph{verifier}. 
The prover has access to the graph (including the structure) and has unlimited power. It assigns the \emph{proof labels} (sometimes called certificates \cite{bousquet_et_al:LIPIcs.DISC.2025.18}) and wants the nodes to accept, irrespective of whether it is a yes-instance or a no-instance. 
The verifier is a distributed algorithm running on every node that distinguishes between yes-instances, where the prover is helping, and no-instances, where the prover is trying to mislead the nodes, with the help of the proof labels of the node and its neighbors. 
Later G{\"o}{\"o}s and Suomela \cite{goos2016locally} generalised the idea of proof labeling schemes under the name of \emph{locally checkable proofs} (LCP), where the main difference is that a node can inspect a constant radius (instead of just 1) neighborhood (often called \emph{view distance}) and that the node can access the proof labels and identifier of the nodes within that ball (formally defined in Sec. \ref{sec:model}). A survey by Feuilloley \cite{feuilloley2021introduction} consolidates these notions. This framework on local verification has been extended in various directions, including PLSs with restricted provers \cite{DBLP:conf/wdag/EmekGK22}, randomized \cite{10.1145/2767386.2767421} and quantum PLSs \cite{fraigniaud_et_al:LIPIcs.ITCS.2021.28}, PLSs with global proofs in addition to local ones \cite{feuilloley_et_al:LIPIcs.DISC.2018.25}, and many more.

In PLS and LCPs, an important metric is the proof size (the number of required bits), and a long line of work focuses on this for various graph properties on different families of graphs \cite{bousquet_et_al:LIPIcs.DISC.2025.18,10.1145/3519270.3538416,fraigniaud_et_al:LIPIcs.DISC.2023.20,DBLP:journals/algorithmica/FraigniaudMRT24}.
Although certain graph properties, e.g. $s-t$ reachability, bipartiteness, and cycles of even length, have been studied in \cite{goos2016locally} with constant-size proofs, a proof size of $O(\log n)$ has emerged as quite common and standard for many local verification problems  \cite{feuilloley_et_al:LIPIcs.DISC.2018.25,10.1145/2499228}, including verifying the existence of a unique leader \cite{feuilloley2025proving}, checking cycle-freeness, spanning tree \cite{10.1007/978-3-319-12340-0_2,ostrovsky2017space}, non-bipartiteness \cite{goos2016locally} and minor-freeness \cite{bousquet2024local}.
In a recent line of works \cite{ESPERET202268,10.1145/3382734.3404505,FEUILLOLEY20239}, it has been shown that planar graphs and graphs of bounded genus can be verified with $O(\log n)$ bits.
Some verification problems require even larger labels: for example, the MST-verification problem \cite{DBLP:journals/dc/KormanK07} (edge weights in $[1, W]$) with a proof of size $O(\log n \log W)$, diameter $k$ verification with an $O(n \log n)$ proof \cite{CENSORHILLEL2020112}.
It is common intuition that hereditary properties (i.e., removing a node or an edge yields a graph that falls in the same class as the original one) are easier to verify (using $O(\log n)$ size labels) in the sense that information about every single node or edge is not required to be encoded in the proof. 

\vspace{1.7mm}
\noindent $\blacktriangleright$ \textbf{Our Model and Results.} In this paper, we adopt the model of locally checkable proofs (LCP), with a more focused concentration on the verification of \emph{cycle existence} -- for a cyclic graph, there must exist a proof such that the verifier makes all nodes accept (output $1$), and for an acyclic graph, the verifier makes some node(s) reject (output $0$) for any possible proof.
This problem remains particularly interesting to us as it can be verified using only $3$ distinct proof labels, i.e., $O(1)$-sized proof, unlike cycle-freeness.
We emphasize that cycle-freeness is different from cycle existence in the context of LCPs, where the problem statement is quite the opposite -- for an acyclic graph, there must exist a proof such that the verifier makes all nodes accept, and for a cyclic graph and for any proof on it, there must exist at least one rejecting node. We divide our contributions into two major verticals: \emph{oracular proofs} (where the prover assigns correct proof labels on the nodes) and \emph{erroneous proofs} (where an adversary can corrupt parts of the oracular proof with a goal of misleading the verifier).
We talk more about each of these in subsequent paragraphs (see also Table \ref{table:results}).

\newcolumntype{C}{>{\centering\arraybackslash}X}

\renewcommand\theadfont{\scriptsize\bfseries}
\renewcommand\theadalign{cc}

\begin{table}[ht]
    \centering
    \scriptsize
    \setlength{\tabcolsep}{2pt} 
    \renewcommand{\arraystretch}{1.5} 

    \begin{tabularx}{\textwidth}{|>{\centering\arraybackslash}p{1.4cm}||c|c|c|C||c|c|c|c|c|C|}
    \hline

    \multirow{2}{*}{\thead{Problem}} & \multicolumn{4}{c||}{\thead{Proofs without Errors}} & \multicolumn{6}{c|}{\thead{Proofs with Errors}} \\ \cline{2-11} 
     & \thead{View \\ Dist.} & \thead{Distinct \\Proof \\ Labels} & \thead{+ \\or \\--} & \thead{Reference} & \thead{No of\\ Errors} & \thead{Error \\ Dist.} & \thead{View \\ Dist.} & \thead{Distinct \\Proof \\ Labels} & \thead{+ \\or \\--} & \thead{Reference} \\ \hline \hline

    \multirow{3.5}{*}{\makecell{Cycle \\ existence}} 
    & 1 & 3 & $+$ & Thm. \ref{thm:cycle_3labels} & \multirow{2}{*}{$i$} & \multirow{2}{*}{any} & \multirow{2}{*}{$i$} & \multirow{2}{*}{$\infty$} & \multirow{2}{*}{$-$} & \multirow{2}{*}{Thm. \ref{thm:impos_cycle}} \\ \cline{2-5}
    & 1 & 2 & -- & Thm. \ref{thm:impos_2labels} & & & & & & \\ \cline{2-11}
    & 3 & 2 & $+$ & Thm. \ref{thm:cycle_2labels_3viewdistance} & $i$ & $2i+1$ & $2i+1$ & $O(n)$ & $+$ & Sec. \ref{subsec:cycle-detection-with-errors} \\ \hline \hline

    \multirow{2}{*}{\makecell{Cycle \\ freeness}} 
    & 1 & $O(n)$ & $+$ & \cite{goos2016locally,ostrovsky2017space} & $i$ & any & $i$ & $\infty$ & $-$ & Sec. \ref{sec:lower-bound-erroneous-labelings} \\ \cline{2-11}
    & $k$ & $O(n)$ & -- & \cite{ostrovsky2017space} & $i$ & $2i+1$ & $2i+1$ & $O(n)$ & $+$ & Sec. \ref{sec:erroneous-labeling} \\ \hline \hline

    \multirow{2}{*}{\makecell{Bipartit- \\ eness}}
    & \multirow{2}{*}{1}
    & \multirow{2}{*}{2}
    & \multirow{2}{*}{$+$}
    & \multirow{2}{*}{trivial}
    & $i$ & any & $i$ & $\infty$ & $-$
    & Sec.~\ref{sec:lower-bound-erroneous-labelings}
    \\ \cline{6-11}
    
    &&&&
    & $i$ & $2i+1$ & $2i+1$ & 2 & $+$
    & Sec.~\ref{sec:erroneous-labeling}
    \\ \hline

    \end{tabularx}
    \caption{Summary of results. Here $+$ denotes achievability and $-$ denotes impossibility. Column ``No of Errors'' stores the number of errors within the distance ``Error Dist.'' from each node.}
    \label{table:results}
    
\end{table}

To our knowledge, the literature mostly discusses proofs, which we have termed oracular proofs in this paper. 
In many verification problems (e.g. cycle-freeness, uniqueness of the leader, cycle existence, path parity, etc.), a common and natural approach is to consider proofs that establish a sense of direction (or orientation) for edges of the graph. Although different problems require different proof construction based on the property we are verifying, a general strategy is to designate a node (or sometimes a connected subgraph, especially for cycle existence, which we will see in Sec. \ref{sec:3labels}) as a root, assign it the proof label $0$, and assign every other node the (BFS) distance from it as the proof label. 
Restricting attention to constant-size proofs—specifically, using only $2$ or $3$ distinct labels—we ask whether it is possible to encode such an orientation.

This is indeed possible, and we show it for cycle existence (Definition \ref{def:cycle-existence}).
Using $3$ distinct proof labels and a $1$-hop view for each node, such an orientation can be defined naturally in this setting. However, we show that reducing the number of proof labels to $2$ under the same constraint is impossible.
What is more interesting in cycle existence is that, with a cost of $3$-hop view distance for every node, we present a non-trivial way in Sec. \ref{sec:2label-3viewdistance-algorithm} to encode orientation using only $2$ labels (using a repeated use of a special string $001101$). 
This technique may be of independent interest for other verification problems, e.g., Orientation in anonymous trees Problem \cite{10.1145/1073814.1073817} using only $2$ labels.

We have another vertical of results, which concerns \emph{proofs with errors}. 
We study a setting in which the adversary is allowed to modify proof labels of at most $i$ nodes within the $(2i+1)$-hop neighborhood of each node (defined formally in Sec. \ref{sec:model}).
To tackle this, we introduce an algorithmic framework \textbf{\texttt{refix}} that requires every node to have a view distance of $2i+1$.
Although we could not prove that \framework{} works for any verification problem or provide a formal classification of verification problems that can be verified using \framework{}, we illustrate its applicability to three verification problems: cycle existence, cycle-freeness, and bipartiteness. 
We establish a lower bound in Sec. \ref{sec:lower-bound-erroneous-labelings}, to demonstrate a correlation between $i$ and the view distance required for every node. The erroneous proof model resembles a prediction-like model, which is a relatively recent topic among the distributed community. Despite some differences with our model, one can refer to Boyar et al. \cite{10.1145/3732772.3733530}, who studied a prediction model on classical graph problems like MIS, maximal matching, and coloring.

Finally, we initiate the study of implementing LCPs in the CONGEST model. We show that the $2$-label, view-distance-$3$ verifier for cycle existence admits a $3$-round CONGEST implementation. Since LCPs are naturally designed around nodes accessing larger local neighborhoods, efficiently conveying such information under bandwidth constraints is a non-trivial challenge.
This suggests a new direction at the intersection of local verification and communication constraints. We conclude with several open problems, including the implementation of LCPs with errors in CONGEST.



\section{Model and Preliminaries}
\label{sec:model}

\noindent $\blacktriangleright$ \textbf{Graphs.}
The graphs under consideration are simple connected graphs. We represent a graph by $G$ with $V(G)$ as the set of nodes and $E(G)$ as the set of edges.
When the context is clear, we omit $G$ from the vertex and edge set.
We denote the number of nodes in $G$ by $n$.
By $dist(u,v)$, we represent the distance, i.e., the length of a shortest path between two nodes $u$ and $v$.
For a node $v$, we use the notation $N_{k}(v)$, called the \emph{$k$-hop neighborhood of $v$}, to denote the set of all nodes within the distance $k$ hops from $v$ (the shortest path from $v$ to any node in $N_k(v)$ has  $\leq k$ edges), including $v$ itself.
The rest of the model components are the same as in \cite{DBLP:conf/podc/GoosS11}.

\vspace{1mm}
\noindent $\blacktriangleright$ \textbf{Local Verification.}
As mentioned earlier, we follow the \emph{locally checkable proof} (LCP) model (ref. \cite{DBLP:conf/podc/GoosS11}) in this paper.
Let us formally define all its components for our paper.

Let us consider a function $L: V \rightarrow \{0,1\}^*$ that maps a binary string to every node of $G$. 
The function $L$ is referred to as \emph{a proof} for $G$.
From the perspective of a node $v$, we call $L(v)$ the \emph{proof label} of $v$.
The maximum number of bits in any proof label $L(v)$ for $v \in V$ is called the size $|L|$ of the proof $L$.

Let $k\geq 1$. A \textit{local verifier} $\mathcal{A}$ is a function  that maps every triple $(G[N_{k}(v)],$ $ L[N_k(v)], v)$ to a binary output $0$ or $1$ for some $k$, where $G[N_k(v)]$ is the subgraph induced by the nodes in $N_k(v)$, $L[N_k(v)]$ is the proof restricted on the subgraph $G[N_k(v)]$ and $v \in V$.
In this case, we say that $v$ has a \emph{view distance} $k$ that consists of all the nodes within the neighborhood $N_k(v)$ along with their proof labels and the graph structure induced by all those nodes.
For every node $v$, the verifier $\mathcal{A}$ takes as input the view of $v$ and outputs either $1$ (accept) or $0$ (reject) based on that information.

Recall that a graph property $\mathcal{P}$ is a subset of graphs that is closed under isomorphism. 
A graph property $\mathcal{P}$ admits a \textit{locally checkable proofs} (in short, LCP) if there is a local verifier $\mathcal{A}$ such that (\emph{\textbf{Completeness}}) if a graph $G \in \mathcal{P}$, there is a proof $L$ of $G$ such that $\mathcal{A}(G[N_{k}(v)],$ $ L[N_k(v)], v) = 1$ for every node $v \in V$, and (\emph{\textbf{Soundness}}) if $G \notin \mathcal{P}$, then for every proof $L$, there exists a node $v \in V$ such that $\mathcal{A}(G[N_{k}(v)], L[N_k(v)], v) = 0$.

We say that the node $v$ \emph{accepts} (resp. \emph{reject}) when it outputs $1$ (resp. $0$).
We put emphasis on the fact that when $G$ does not satisfy the property $\mathcal{P}$, there must exist at least one rejecting node for any proof assignment. 
It is also well-accepted to refer to a local verifier $\mathcal{A}$ as  
a \textit{verification algorithm}, and we adopt the same. 

\textbf{Equivalence with the LOCAL Model:} As mentioned earlier, this model is equivalent to the well-known LOCAL model \cite{peleg2000distributed}, in the sense that a verifier with view distance $k$ can be implemented in $k$ synchronous communication rounds \cite{DBLP:conf/podc/GoosS11} of the LOCAL model.

\vspace{1mm}
\noindent $\blacktriangleright$ \textbf{Erroneous or Adversarial Proofs.}
 In this paper, we segregate the concept of proofs on a given graph $G$ into two types. The first one is an \emph{oracular proof}, denoted by $L_O$ and provided by a trustworthy oracle, which means that the completeness and the soundness properties are satisfied under $L_O$. 
The other one is referred to as the \emph{erroneous} or \emph{adversarial proof}, denoted by $L_{adv}$, which is provided by the adversary by modifying the oracular proof of the graph $G$ for some property $\mathcal{P}$.
Such a modification introduced by the adversary may potentially cause all nodes to accept even if $G \notin \mathcal{P}$ or some nodes to reject when $G \in \mathcal{P}$.

To study its effect in the context of verification problems, we consider the following adversarial proof model.
In this paper, we assume that $L_{adv}$ differs from $L_O$ in at most $i$ nodes within $N_{2i+1}(v)$ for every node $v \in V$.
A graph property $\mathcal{P}$ admits an \textit{erroneous proof} if there is a local verifier $\mathcal{A}$ satisfying the following two properties:

\begin{itemize}
   \item \textit{Completeness:} If a given graph $G$ is in $\mathcal{P}$, there is a proof $L_O$ of $G$ such that for every $L_{adv}$ obtained from $L_O$ (i.e., $L_{adv}$ differs from $L_O$ in at most $i$ nodes within $N_{2i+1}(v)$ for every node $v \in V$) we have $\mathcal{A}(G[N_{k}(v)],$ $ L_{adv}[N_k(v)], v) = 1$ for every node $v \in V$.

    \item
    \textit{Soundness:} If $G$ is not in $\mathcal{P}$, then for every proof $L$, there exists a node $v \in V$ such that $\mathcal{A}(G[N_{k}(v)], L[N_k(v)], v) = 0$.
    
\end{itemize}

Notice that the completeness property has now changed compared to the model without errors. At an intuitive level, this means that the adversary modifies oracular proofs in such a way that for each node $v$, there are at most $i$ nodes within $N_{2i+1}(v)$, whose proof labels have been modified by the adversary, yet the local verifier is still able to accurately detect property $\mathcal{P}$.
We address each such modification of proofs as \emph{an error} introduced by the adversary.
From now on, we use $i$ exclusively to denote the number of errors.


\section{Local Verification with Oracular Proofs}
\label{sec:warm-up}

We begin with oracular proofs for three verification problems: cycle existence, cycle-freeness, and bipartiteness. 
Our main focus is on verifying cycle existence, as we show in this section some interesting results on this. 
We reemphasize that cycle existence and cycle-freeness are different (ref. following definitions). 
Cycle-freeness has already been studied in the literature~\cite{goos2016locally,ostrovsky2017space}, and Bipartiteness admits a trivial LCP. 
We discuss these LCPs further in more detail when we examine the proofs with errors later in this paper.

\begin{definition}{(Cycle Existence):}
\label{def:cycle-existence}
    Cycle existence is the following verification problem: if $G$ has a cycle, there must exist a proof such that a local verifier $\mathcal{A}$ makes all nodes accept; if $G$ is acyclic, then for every proof, $\mathcal{A}$ makes at least one node reject.
\end{definition}

\begin{definition}{(Cycle-freeness):}
    If $G$ is acyclic, there must exist a proof such that a local verifier $\mathcal{A}$ makes all nodes accept; if $G$ has a cycle, then for every proof, $\mathcal{A}$ makes at least one node reject.
\end{definition}

\subsection{Cycle Existence: View distance 1 and 3 Distinct Proof Labels}
\label{sec:3labels}

Let us first define the notions of core and tree nodes for convenience in the description. 
\begin{definition}{(Core Node and Tree Node):}
\label{definition:core-tree-node}
     A node $v$ is called a core node if it lies on a cycle or on any path connecting two cycles. Otherwise, it is called a tree node.
\end{definition}

Consider the following proof $L_O$ (with $3$ distinct proof labels $0,1,2$), illustrated in Fig. \ref{fig:4labels}.

\begin{itemize}
    \item If $v$ is a core node, $L_{O}(v) = 0$.

   \item Otherwise, $L_{O}(v)$ is equal to the distance to the nearest node with the proof label $0$ (i.e., a node on a cycle or on a path between cycles) taken modulo $3$.
\end{itemize}

\begin{wrapfigure}[11]{r}{0.55\textwidth}
  \centering
    \includegraphics[scale=0.45]{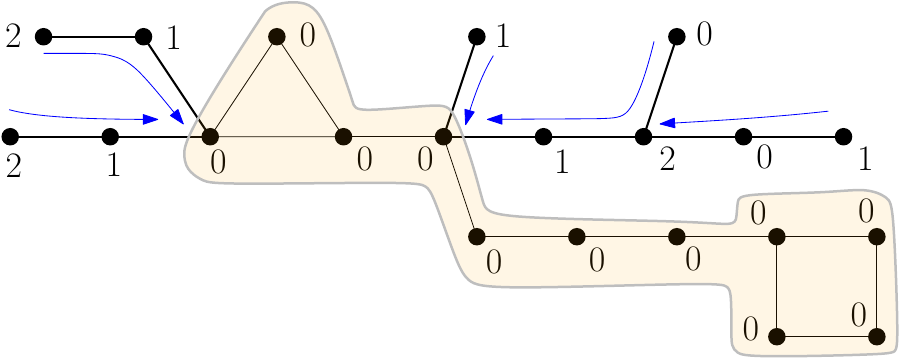}
    \caption{Cycle existence proof $L_O$ with $3$ proof labels. The implied direction among tree nodes is also marked.
    }

    \label{fig:4labels}
\end{wrapfigure}
\noindent $\blacktriangleright$ \textbf{Parent-Child Orientation by $L_O$:} 
Under this proof, for a tree node $v$, the neighboring node $u$ (resp. $w$) with the proof label $L_O(u) = L_O(v) - 1\pmod{3}$ (resp. $L_O(w) = L_O(v) + 1\pmod{3}$) is referred to as \textit{the parent} (resp. \textit{a child}) of $v$.
This parent-child relation between a pair of neighboring nodes establishes an orientation for every edge.

\noindent $\blacktriangleright$ \textbf{Verification Algorithm:} We now move to a quite straightforward algorithm $\mathcal{A}\textsc{-3labels}$ executed at node $v$. .
Let $L(v)$ denote the proof label of a node $v$.
If $L(v)=0$, then $v$ accepts if it has at least two neighbors with the proof label $0$ and no neighbor with proof label $2$.
Otherwise, when $L(v)\neq 0$, it checks whether exactly one of its neighbors has the proof label $L(v)-1 \pmod{3}$, while every other has the proof label $L(v)+1 \pmod{3}$. If so, $v$ accepts.
If neither of the above conditions is satisfied,  $v$ rejects.

\begin{algorithm2e}
    $L(v) \gets$ The proof label of $v$

    \If{$L(v) = 0$ and $v$ has at least two neighbors with the proof label $0$, but it has no neighbor with the proof label $2$}
        {
            $v$ accepts
        }
    \ElseIf{$L(v)\neq 0$ and $v$ has exactly one neighbor with the proof label $L(v)-1\pmod{3}$ and all other neighbors with the proof label $L(v)+1\pmod{3}$\label{line:labeling-on-tree-parts}}
        {
            $v$ accepts
        }
    \Else   
        {
            $v$ rejects
        }

	\caption{$\mathcal{A}\textsc{-3labels}$ (for node $v$)}
    
 \label{3label-algorithm-pseudocode}
\end{algorithm2e}

We prove the following theorem.

\begin{theorem}
\label{thm:cycle_3labels}
    The proof $L_O$, consisting of $3$ distinct proof labels, together with the verification algorithm $\mathcal{A}\textsc{-3labels}$, establishes an LCP for verifying cycle existence with every node having view distance $1$. 
\end{theorem}

\begin{proof}[Proof of Theorem \ref{thm:cycle_3labels}] 

We prove the two properties, completeness and soundness, separately. 

Assume that $G$ contains at least one cycle. We show that $L_{O}$ is the proof on $G$ such that every node accepts according to Algorithm $\mathcal{A}\textsc{-3labels}$. 
Let $v$ be an arbitrary node in $G$. 
If $L_{O}(v)=0$, then $v$ lies either on a cycle or on a path connecting cycles. In both cases, $v$ has at least two neighbors with the proof label $0$, and by construction, none of its neighbors has the proof label $2$, so $v$ accepts. 
If $L_{O}(v)\neq 0$, then $v$ has exactly one neighbor with $L_{O}(v)-1$ and all other neighbors with $L_{O}(v)+1$.
Therefore, $v$ accepts.

We now show that for an acyclic graph $G$, no proof on $G$ can be accepted by all nodes. 
Let $L$ be an arbitrary proof on $G$.
We first consider a node with $v$ with $L(v) = 0$ (if it exists).
Consider a maximal connected path $P$ of nodes labeled $0$ that contains $v$.
Let $u$ be an endpoint of $P$.
Since the graph does not have any cycle, $u$ has at most one neighbor with the proof label $0$, and thus it rejects according to $\mathcal{A}\textsc{-3labels}$. 
Now we consider the case where none of the nodes has the proof label $0$. 
In this situation, we select an arbitrary node $v_0$ in $G$. Node $v_0$ rejects if it does not have a neighbor with label $L(v_0) - 1\pmod{3}$ (i.e., a parent under $L$), and we are done. 
Otherwise, let $v_1$ be the parent of $v_0$ such that $L(v_1) = L(v_0) -1$. 
Applying the same reasoning repeatedly on nodes $v_1$ and subsequently on its parent $v_2$, we obtain a sequence of nodes $v_0, v_1, v_2, \cdots$.
Since $G$ is acyclic, this sequence must eventually terminate at some node $v_k$ that does not have a parent, i.e., a node with the proof label $L(v_k)-1$. 
Such a node necessarily rejects according to $\mathcal{A}\textsc{-3labels}$.
This completes the proof of the theorem. 
\end{proof}

\subsection{Impossibility on Cycle Existence: View Distance 1 \& 2 Proof Labels}
\label{sec:2label-impossibility}

The LCP with $3$ distinct proof labels raises the question of whether $2$ distinct proof labels suffice with a view distance of $1$. The following theorem rules this out.

\begin{theorem}
    \label{thm:impos_2labels}
    There is no algorithm $\mathcal{A}$ that together with a proof $L$ consisting of only $2$ distinct proof labels (i.e., $|L| = 1$) establishes an LCP for cycle existence with the assumption that every node has the view distance $1$.
\end{theorem}
\begin{figure}
\begin{minipage}[b]{0.32\linewidth}
    \centering
    \includegraphics[scale=0.45]{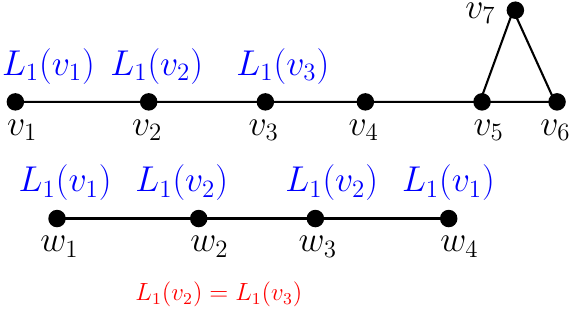}

\end{minipage}\hfill
\begin{minipage}[b]{0.32\linewidth}
    \centering
\includegraphics[scale=0.45]{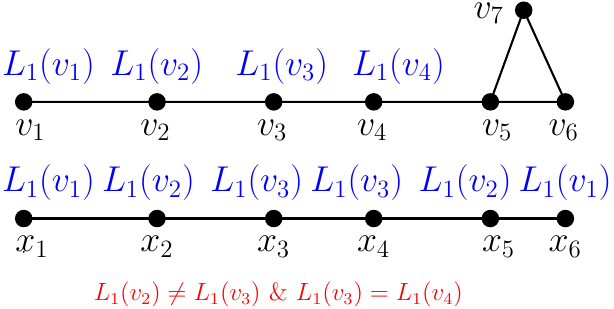}

\end{minipage}\hfill
\begin{minipage}[b]{0.32\linewidth}
    \centering
\includegraphics[scale=0.45]{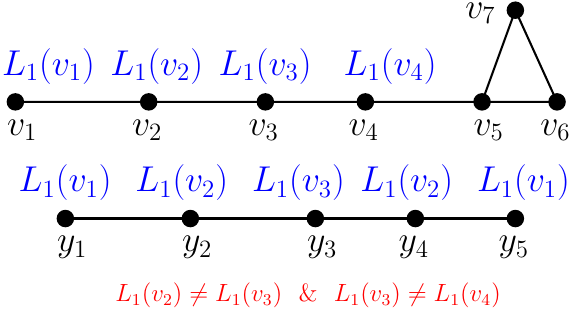}

\end{minipage}
\caption{In Theorem \ref{thm:impos_2labels}, we have only $2$ distinct proof labels. The top graph is $G_1$ (in all three sub-figures), while the (bottom) left, middle, and right graphs are $G_2$, $G_3$, and $G_4$, respectively.}

\label{fig:2-labels-impossibility}   
\end{figure}

\begin{proof}[Proof of Theorem \ref{thm:impos_2labels}]
Assume, for the sake of contradiction, that there exists a proof $L_1$ and a verification algorithm $\mathcal{A}$ for verifying cycle existence that requires only $2$ distinct proof labels and view distance $1$ for every node in the given graph.
Let us consider a cyclic graph $G_1$ consisting of $7$ nodes $\{v_1,v_2,\cdots,v_7\}$ connected in a path, with an additional edge $(v_5,v_7)$ creating a cycle.
We show that if all nodes in $G_1$ accept with the help of the proof $L_1$ and $\mathcal{A}$, then the algorithm $\mathcal{A}$ makes all nodes in an acyclic graph accept for some proof (we will construct such proofs from $L_1$), which is a violation of the soundness property.

We consider the following three cases depending on the proof labels under $L_1$.

$\blacktriangleright$ \textbf{Case i ($L_1(v_2) = L_1(v_3)$):}
Let $G_2$ be a simple path on four nodes $w_1,w_2,w_3,w_4$ with the proof $L_2$ defined as follows, see Fig. \ref{fig:2-labels-impossibility} (left).
$L_2(w_1)=L_2(w_4)=L_1(v_1)$ and
$L_2(w_2)=L_2(w_3)=L_1(v_2)$.
Both nodes $w_1$ and $w_4$ of $G_2$ under the proof $L_2$ have the same local view as node $v_1$ under $L_1$.
Since $\mathcal{A}$ makes the node $v_1$ accept, the same verification algorithm makes $w_1$ and $w_4$ also accept.
Similarly, nodes $w_2$ and $w_3$ under the proof $L_2$ have the same local view as node $v_2$ under $L_1$, since $L_1(v_2)=L_1(v_3)$.
Thus, all nodes in $G_2$ accept, violating the soundness property.

$\blacktriangleright$ \textbf{Case ii ($L_1(v_2) \neq L_1(v_3)$ and $L_1(v_3) = L_1(v_4)$):}
In this case, we consider another acyclic graph $G_3$ consisting of six nodes $x_1,x_2,x_3,$ $x_4,x_5,x_6$ arranged in a simple path.
Define a proof $L_3$ on $G_3$ by setting
$L_3(x_1)=L_3(x_6)=L_1(v_1)$,
$L_3(x_2)=L_3(x_5)=L_1(v_2)$,
and
$L_3(x_3)=L_3(x_4)=L_1(v_3)$, as shown in Fig. \ref{fig:2-labels-impossibility} (middle).

Nodes $x_1$ and $x_6$ under the proof $L_3$ have the same local view as node $v_1$ under $L_1$, and hence both accept.
Similarly, nodes $x_2$ and $x_5$ have the same local information as node $v_2$, and thus accept as well.
Finally, nodes $x_3$ and $x_4$ have the same local view as node $v_3$, since $L_1(v_3)=L_1(v_4)$, and therefore also accept.
Hence, all nodes in $G_3$ accept under $L_3$, violating the soundness~property.

$\blacktriangleright$ \textbf{Case iii ($L_1(v_2) \neq L_1(v_3)$ and $L_1(v_3) \neq L_1(v_4)$):}
In this case, we consider another acyclic graph $G_4$, a path on five nodes $y_1, y_2, y_3, y_4,$ and $y_5$, as shown in Fig. \ref{fig:2-labels-impossibility}~(right).
Clearly $L_1(v_2)=L_1(v_4)$ since there are only $2$ distinct proof labels available.
Define a proof $L_4$ on $G_4$ by setting
$L_4(y_1)=L_4(y_5)=L_1(v_1)$,
$L_4(y_2)=L_4(y_4)=L_1(v_2)$,
and
$L_4(y_3)=L_1(v_3)$, where $L_1(v_3)\neq L_1(v_2), L_1(v_4)$.

 The nodes $y_1$ and $y_5$ have the same local view as node $v_1$, and hence accept.
Similarly, $y_2$ and $y_4$ have the same local information as node $v_2$, and therefore accept.
Finally, node $y_3$ has the same local information as node $v_3$, since $L_1(v_2)=L_1(v_4)$, and thus also accepts.
Therefore, all nodes in $G_4$ accept under $L_4$, again violating soundness.

The above three cases together prove Theorem \ref{thm:impos_2labels}.
\end{proof}

\subsection{Cycle Existence:  2 Distinct Proof Labels but with View Distance 3}
\label{sec:2label-3viewdistance-algorithm}

We can circumvent the impossibility (Theorem \ref{thm:impos_2labels}) for cycle existence by considering an increased view distance for each node.
Observe that the proof $L_0$ in Sec. \ref{sec:3labels} establishes a sense of direction in the acyclic components (the ones we get after deleting all core nodes) of the graph along any path from a leaf node toward a cycle (by repeated use of $1, 2, 0$ (in order) on that path, see Fig. \ref{fig:4labels}). 
When a node can see only its $1$-hop neighborhood, a sense of direction cannot be created with just $2$ distinct proof labels (due to symmetry).
However, it is interesting to see that an increment in the view distance (to $3$) enables us to construct an interesting proof $L^*_O$ (with  $2$ distinct proof labels), see Fig. \ref{fig:2labels}.
\begin{figure}[h]
    \centering
    \includegraphics[width=0.5\linewidth]{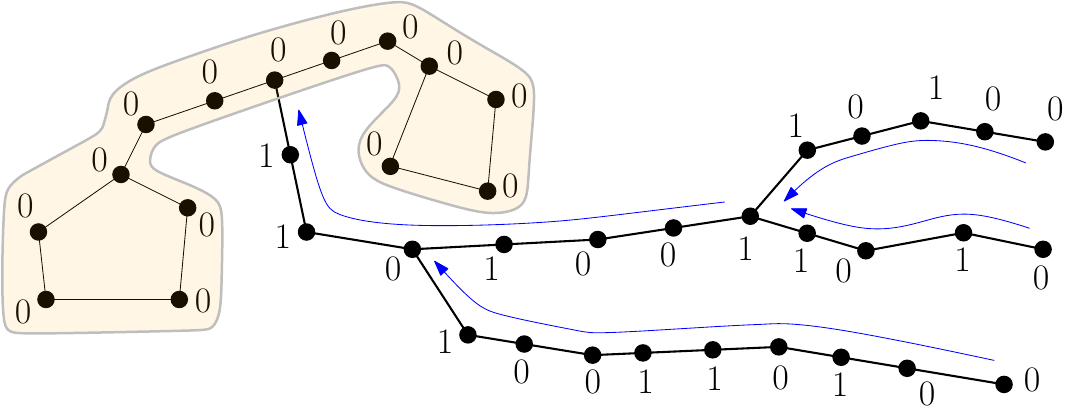}
    \caption{The proof $L_O^*$ with $2$ distinct proof labels, where each node has the view distance $3$. 
    }

    \label{fig:2labels}
\end{figure}

\begin{itemize}
    \item 
    If $v$ is a core node, $L^*_O(v) = 0$. 

    \item 
    Otherwise, let $u$ be the nearest core node to $v$. 
    If $dist(v,u) = 1, 2, 4\pmod{6}$, we have $L^*_O(v) = 1$.
    If instead, $dist(v, u) = 3, 5, 6\pmod{6}$, we have $L^*_O(v) = 0$.
\end{itemize}

Notice that all core nodes create a single connected component. Once we remove it from the graph, the remaining nodes create a forest. 
Before moving on to the algorithm description, let us define some notation for future use.
\begin{itemize}
    \item For two binary strings $S$ and $S'$, we use $S' \sqsubseteq S$ to represent that $S'$ is a substring of $S$.

\item A \emph{route} is a walk $P=v_1v_2\dots v_x$ for some $x$ such that $v_i \neq v_{i+2}$ for all $i$. This means that routes do not permit walking back and forth between two consecutive nodes in the walk, but they do allow walking repeatedly along a cycle.

\item By \emph{$v$-centered route}, we refer to a route $P$ consisting of an odd number of nodes on it such that $v$ is the central node in $P$. For example, refer to the highlighted path in Fig. \ref{fig:string4-2labels}.
\end{itemize}

\noindent $\blacktriangleright$ \textbf{General Idea of the Algorithm:} If a graph has a cycle, to verify cycle existence, a sense of direction among the tree nodes can be encoded using only $2$ distinct proof labels, as long as a node $v$ can see a $v$-centered route $P$ of $5$ nodes with their respective proof labels.
We separately handle the cases when $v$ is a leaf node or a tree node neighboring a leaf node (i.e. when $v$ is not a part of a $v$-centered route, consisting of $5$ nodes).

The trick lies in the repeated use of the binary bit-string $S = 001101$ cyclically (each bit of the string corresponds to a proof label of a node, as depicted in Fig.~\ref{fig:2labels}).
Consider the string $S_\infty=001101001101001101\cdots$. We call the $6$ distinct length-$5$ substrings of $S_{\infty}$ the \emph{base strings}; they are listed in Table \ref{table:refined-membership}.
We denote a base string by $S_b$. 
An important property of $S_\infty$ is that no reverse of a base string is itself a substring of $S_\infty$.

One can refer to Fig. \ref{fig:string4-2labels} to get an idea about how a base string is correlated with a $v$-centered route consisting of $5$ nodes (the highlighted path). 
For the sake of simplicity, when we say that \textit{a node $v$ sees a length-$5$ string within $N_2(v)$} (not necessarily a base string only) if the proof labels of the nodes on a $v$-centered route of $5$ nodes correspond to the length-$5$ string. For example, in Fig. \ref{fig:string4-2labels}, $v$ sees the string $10100$, and in Fig. \ref{fig:string3-2labels}, $v$ sees the string $00100$.

We use these base strings to our advantage, as each of them defines a sense of direction for the central node. For example, let us consider the base string $10100$, which corresponds to a route of $5$ nodes, where the proof label of the central node $v$ is the third bit $1$ in the base string (see Fig. \ref{fig:string4-2labels}). 
Since none of the base strings is a palindrome, upon seeing $10100$, the node $v$ determines that the second bit $0$ corresponds to the proof label of a node, which is its parent. Similarly, the fourth bit of $10100$ corresponds to the proof label of a child of $v$. 
Notice that by definition, any base string corresponds to a $v$-centered route of $5$ nodes.
However, $v$ can be a part of multiple $v$-centered routes within $N_2(v)$, and it is not necessary that each such route corresponds to some base string. Hence it can see multiple length-$5$ strings within $N_2(v)$ (e.g., in Fig. \ref{fig:string3-2labels}, $v$ sees $00100$). 
Let us classify them as follows, and later, we highlight why this classification is beneficial for the verification algorithm. 


\vspace{1mm}
\textbf{Tree String Set:}
When the given graph contains a cycle, the (oracular) proof $L_O^*$ assigns proof labels to the tree nodes (recall Definition \ref{definition:core-tree-node}) in such a way that each tree node $v$ must see a base string $S_b$, thereby establishing a sense of direction (i.e., enabling $v$ to identify its parent and children). 
Let $v_{-2} v_{-1} v v_1 v_2$ be the $v$-centered route on which $v$ sees the base string $S_b$, and $v_{-1}$ is the parent of $v$ according to $S_b$.
Corresponding to the sense of direction the base string $S_b$ has established, any other length-$5$ string (associated with some $v$-centered route of $5$ nodes) should be one of the following types.

\begin{figure}[t]
\begin{minipage}[b]{0.32\linewidth}
    \centering
\includegraphics[width=0.8\linewidth]{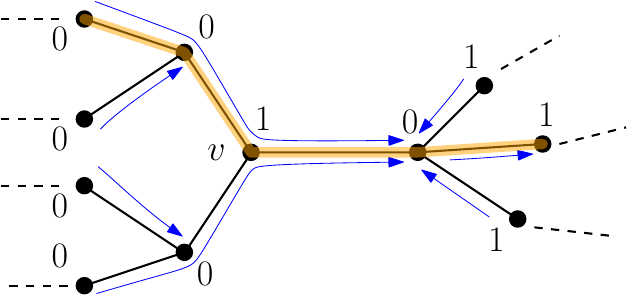}
    \caption{Highlighted route corresponds to tree strings of the form $l_{-2} l_{-1} l_0 l_{1} l_{2}$ (=$10100$ here)}
    
   \label{fig:string4-2labels}
\end{minipage}\hfill
\begin{minipage}[b]{0.32\linewidth}
    \centering
    \includegraphics[width=0.8\linewidth]{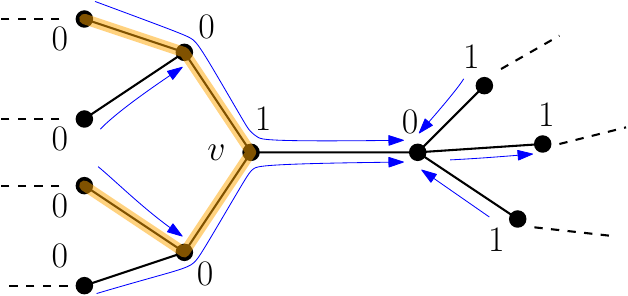}
    \caption{Highlighted route corresponds to a tree string of the form $l_{-2} l_{-1} l_0 l_{-1} l_{-2}$ (= $00100$ here)}
    
    \label{fig:string3-2labels}
\end{minipage}\hfill
\begin{minipage}[b]{0.32\linewidth}
    \centering
\includegraphics[width=0.8\linewidth]{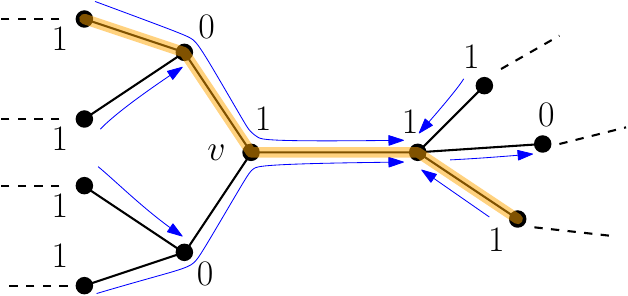}
    \caption{Highlighted route corresponds to tree strings of the form $ l_{0} l_{1} l_0 l_{-1} l_{-2} $ (=$11101$ here)}
    
    \label{fig:string5-2labels}
\end{minipage}
\end{figure}

\begin{enumerate}
    \item \textbf{\textsc{gc-gp string} (abbreviate as grandchild to grandparent string):} This string corresponds to any grandchild to grandparent (according to the direction established by $S_b$) $v$-centered route. One such route is shown in Fig. \ref{fig:string4-2labels}. We denote this as $ l_{-2}l_{-1}l_0l_1l_2$, where the central bit $l_0$ is the proof label of $v$, and $l_{-1}$ (resp. $l_1$) is the proof label of (any) child (resp. the parent $v_1$) of $v$.
    The bit $l_{-2}$ (resp. $l_2$) is the proof label of (any) grandchild (resp. grandparent).
    Needless to mention that a \textsc{gc-gp string} is the base string itself.
    If the route from the grandchild to grandparent is regarded in the reverse direction, we can simply use the reverse of $l_{-2}l_{-1}l_0l_1l_2$ to represent that.
    It is simply a matter of a representative issue and has no effect on the decision algorithm for the node.

    
    \item \textbf{\textsc{gc-gc} string (abbreviated as grandchild to grandchild string):}
    We denote this as $ l_{-2}l_{-1}l_0l_{-1}l_{-2}$. 
    This string corresponds to any $v$-centered route of $5$ nodes from a grandchild to a grandchild (according to the direction established by $S_b$) of $v$. One such route is shown in Fig.~\ref{fig:string3-2labels}.
    The reverse of this type of string falls into this category itself.

    \item \textbf{\textsc{gc-cp string} (abbreviated as grandchild to child of parent string):} We represent this as $ l_{-2}l_{-1}l_0l_{1}l_{0}$, which corresponds to a $v$-centered route of $5$ nodes from a grandchild of $v$ to another child ($\neq v$) of its parent, see Fig.~\ref{fig:string5-2labels}.
    If the route is regarded in the reverse direction, i.e., from a child of the parent to a grandchild, we reverse the above string.

\end{enumerate}

For a base string $S_b$, the \emph{tree string set} $t(S_b)$ consists of any of the above-mentioned types of length-$5$ string (including their reverses). For each base string, Table \ref{table:refined-membership} lists all the tree string sets (in the second column).

\begin{table}[ht]
\footnotesize
\centering

\resizebox{\textwidth}{!}{
\begin{tabular}{|c|l|p{1.6cm}|p{1.9cm}|p{2.3cm}|}
\hline

\textbf{Base String} & \textbf{Tree String Set $t(S_{bi})$} & \textbf{Which $S_{bj}$$\in~$$t(S_{bi})$?} & \textbf{Central Node Label} & \textbf{Label of parent (2nd Bit)} \\ \hline \hline
$S_{b1}=00110$ & $\{00110, 01100, 01110, 01101, 10110\}$ & $S_{b2}$ & 1 & 0\\ \hline
$S_{b2}=01101$ & $\{01101, 10110, 10101, 10111, 11101\}$ & None & 1 & 1 \\ \hline
$S_{b3}=11010$ & $\{11010, 01011, 01010\}$ & None & 0 & 1 \\ \hline
$S_{b4}=10100$ & $\{10100, 00101, 00100\}$ & None & 1 & 0 \\ \hline
$S_{b5}=01001$ & $\{01001, 10010, 10001\}$ & None & 0 & 1\\ \hline
$S_{b6}=10011$ & $\{10011, 11001, 11011, 11000, 00011\}$ & None & 0 & 0\\ \hline
\end{tabular}
}
\caption{The 6 base strings and their respective tree string sets. The third column indicates whether a base string appears in any other tree string set or not.}
\label{table:refined-membership}
\end{table}

Let us now present a scenario to provide intuitive importance of the above classification.

\textbf{A Special Scenario:}
One can ask: \textit{Can $v$ see two different base strings pointing in two different directions, as depicted in Fig \ref{fig:conflicting-base-strings}?} If so, which neighboring node is its parent?

Such a scenario might appear for a node, and that is where the above classification becomes useful. Consider the base string $00110$, as shown in Fig. \ref{fig:conflicting-base-strings}. The intended direction is highlighted with the blue arrow.  
According to this string, the parent of $v$ is $v_2$, and the node $v_4$ is one of the children of $v$. 
However, another $v$-centered route 
$v_5v_4vv_2v_1'$ 
induces a different base string $01101$, which defines $v_4$ as the parent of $v$. 
Observe from Table \ref{table:refined-membership} that $t(00110) = \{00110, 01100, 01110, 10110, 01101\}$. 
We also have $t(01101) = \{01101, 10110, 10101, 10111, 11101\}$.
Clearly, $00110 \notin t(01101)$.
Hence, the node 
$v$ disregards a length $5$ string $S_b$ (in this example $01101$) as the direction-defining base string, if $v$ sees a length-$5$ string within $N_2(v)$ not in $t(S_b)$.

\begin{figure}
    \centering
    \includegraphics[width=0.5\linewidth]{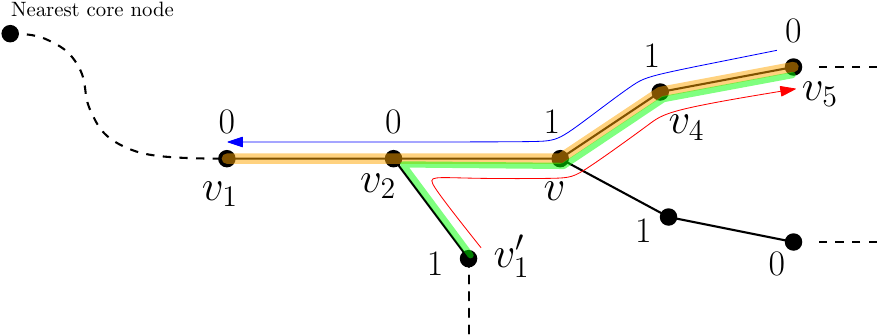}
    \caption{$v$ sees two base strings inducing conflicting directions from two different $v$-centered routes. The direction indicated by the red arrow is discarded since $00110 \notin t(01100)$.
    }
    
    \label{fig:conflicting-base-strings}
\end{figure}

\vspace{1mm}
\textbf{\textsc{Discover-Parent}($v$) Subroutine:} (Algorithm \ref{alg:parent})
Let us handle the easier cases first. 
If $v$ is a leaf, its only neighbor is its parent. 
If $v$ is not a leaf, but it is a tree node and has exactly one neighbor $p$ with degree more than $1$, the node $p$ is its parent.
However, when $v$ is a tree node lying on a $v$-centered route of $5$ nodes, it finds at least one base string under the oracular proof $L_O^*$. 

When $v$ sees a base string $S_b$, rather than immediately taking the decision of acceptance, $v$ performs an additional verification step to check whether there exists any length-$5$ string within its $N_2(v)$ corresponding to some $v$-centered route that does not belong to $t(S_b)$.
If $v$ finds a length-$5$ string within $N_2(v)$ that is not in $t(S_b)$, it disregards $S_b$ and considers another base string.
Table \ref{table:refined-membership} shows which base strings belong to the tree string sets associated with other base strings. However, there is no two base strings $S_{bi}$ and $S_{bj}$ for which both $S_{bi} \in t(S_{bj})$ and $S_{bj} \in t(S_{bi})$ hold true. 
Once verification is complete and $v$ decides on the base string that establishes a direction, $v$ simply identifies the node corresponding to the second bit in that base string as its parent (see the fifth column in Table \ref{table:refined-membership}). 

When $v$ is not a tree node, it tries to determine if it is a core node or not. 
Notice that, in the oracular proof $L_O^*$, the string $000$ (made by the proof labels of three consecutive nodes) can only appear around a core node and does not appear in any base strings as a substring. 
Thus, when $v$ with proof label $0$ finds two neighbors with the proof label $0$, it finds itself as a core node.
If $v$ is neither a tree node nor a core node, it simply rejects at this stage.

\textbf{Verification Algorithm $\mathcal{A}_0$--\textsc{2labels}:} (Algorithm \ref{alg:A_0-algorithm-2labels}) Once the subroutine \textsc{Discover-Parent}() is over, it remains to verify that these local decisions are consistent across neighboring nodes. 
In particular, a core node must be adjacent to at least two other core nodes, and a tree node $v$ must have exactly one parent $p$ such that, from the perspective of $p$, node $v$ is identified as a child of $p$. 
To perform this verification, a node $v$ uses view distance $3$ to simulate the \textsc{Discover-Parent}($u$) algorithm for each neighbor $u \in N_1(v)$ and checks that the resulting decisions agree.

\begin{algorithm2e}[ht]

\If{$v$ is a leaf\label{line:tree_node_start}}
    {
        $v$ returns its only neighbor 
        
    }
\ElseIf{there is exactly $1$ neighbor $p$ of $v$ with $deg(p)>1$}
    {
        $v$ returns $p$ 
        \label{line:tree_node_mid}
    }
\ElseIf{$v$ sees a base string $S_b \sqsubseteq S_\infty$}
    {
        
        \ForEach{$S_b$ within $N_2(v)$}
        {
            Compute $t(S_b)$

            \If{there exists a length $5$ string $S'$ corresponding to a $v$-centered route such that $S' \notin t(S_b)$}
                {
                    continue
                }
            \Else
                {
                    $v$ returns the node $p$ that corresponds to the $2$nd label in string $S_b$ \label{line:tree_node_end}
                }
        }

    }
\ElseIf{$L(v)=0$ and there are two neighbors $u,w$ of $v$ such that $L(u)=L(w)=0$}
    {
        $v$ returns that its a core node
    }
\Else
    {
        $v$ rejects
    }

	\caption{\textsc{Discover-Parent}($v$)}
    \label{alg:parent}

\end{algorithm2e}

\begin{algorithm2e}[ht]

$result \leftarrow $ \textsc{Discover-Parent}$(v)$

\ForEach{$u \in N_1(v)$}
    {
        $result[u] \leftarrow $ \textsc{Discover-Parent}$(u)$
    }

\If{$result = core\ node$ and for at least $2$ nodes $u \in N_1(v)$ we have $result[u] = core\ node$}
    {
        \label{2label-cycle-existence-line-start}
        $v$ accepts
    }
\ElseIf{$result = my\ parent\ is\ p$ \textsc{and} $(result[p] \neq v$ \textsc{or} $result[p] = core\ node)$}
    {
        $v$ accepts\label{2label-cycle-existence-line-end}
    }
    
\Else
    {
        $v$ rejects
    }

\caption{$\mathcal{A}_0$--\textsc{2labels} for node $v$}
\label{alg:A_0-algorithm-2labels}

\end{algorithm2e}

\vspace{1mm}
\noindent $\blacktriangleright$ \textbf{Analysis:}
TWe establish the following theorem.

\begin{theorem}
\label{thm:cycle_2labels_3viewdistance}
    The proof $L_O^*$, consisting of $2$ distinct proof labels, together with the verification algorithm $\mathcal{A}_0$-\textsc{2Labels}, establishes an LCP for verifying cycle existence with every node having view distance $3$.  
\end{theorem}

\begin{proof}[Proof of Theorem~\ref{thm:cycle_2labels_3viewdistance}]
    We prove completeness and soundness by analyzing cyclic and acyclic graphs separately. Before that, the following lemma first proves the correctness of the subroutine \textsc{Discover-Parent}($v$), which helps us establish the completeness property.

    \begin{lemma}
    \label{lem:correctness-algorithm2}
        Under $L_O^*$, every node $v$ correctly determines whether it is a core node or a tree node using \textsc{Discover-Parent}($v$), and in case of $v$ being a tree node, it identifies its parent.
    \end{lemma}

    \begin{proof}
        Let us take an arbitrary node $v$. We distinguish the following two cases depending on whether $v$ is a tree node or a core node.

        $\blacktriangleright$ \textbf{Case 1 ($v$ is a tree node):} We show that \textsc{Discover-Parent}($v$) has node $v$ return its parent. 
        Let us deal with the two cases (corresponding to Lines \ref{line:tree_node_start}-\ref{line:tree_node_mid} of Algorithm \ref{alg:parent}), which are quite immediate, where node $v$ is not part of some $v$-centered route of $5$ nodes.
        If $v$ is a leaf node, the parent of $v$ is the only neighbor of $v$, and hence it is identified as the parent by \textsc{Discover-Parent}($v$).
        If instead, when $v$ is not a leaf node, but has exactly one neighbor $p$ such that $deg(p) > 1$, node $p$ is the parent of $v$ by \textsc{Discover-Parent}($v$). Hence, in both situations, the lemma holds. 
        
        We now concentrate on a more prominent case where $v$ sees a base string $S' \sqsubseteq S_\infty$. 
        Let $u$ be the nearest core node to $v$ and $dist(v,u) = k$. 
        We prove the following three claims.

        \begin{claim}
        \label{claim:unique-u}
         For any $w \in N_2(v)$ with $w \neq v$, node $u$ is the nearest core node to $w$, given the fact that $u$ is the nearest core node to $v$.
             
        \end{claim}

        \begin{proof}
            
        Assume for contradiction that there is a core node $u' \neq u$ such that $dist(w, u') \leq dist(w, u)$.
        Without loss of generality, we may assume that $dist(w,u) \geq dist(v,u)$. 
        If a path $P$ from $u'$ to $u$ contains node $v$, then all nodes on $P$ must be core nodes, contradicting the assumption that $v$ is a tree node. 
        If instead, the path $P$ from $u'$ to $u$ does not contain $v$, both the paths from $v$ to $u'$ and from $v$ to $u$ must pass through $w$. 
        Recall that all core nodes in $G$ form a single connected component and hence $u$ and $u'$ are connected via a path of core nodes. This path, together with the two paths, one from $w$ to $u'$ and the other from $w$ to $u$, forms a cycle. 
        This indicates that $w$ is a core node, contradicting the choice of $u$ as the core node nearest to $v$. This proves Claim \ref{claim:unique-u}.
        \end{proof}
    In continuation of the above discussion, we now prove the following claim.

    \begin{claim}
    \label{claim:unique-parent}
     There is exactly one node $x \in N_2(v)$ with $dist(x, u) = k-1$.
     \end{claim}

    \begin{proof}
        
    Assume to the contrary that there is another node $x'\in N_2(v)$ with $dist(x', u) = k-1$ and $x' \neq x$.
    Consider the two paths from $u$ to $v$: one passing through the node $x$, and the other passing through node $x'$.
    If these two paths are edge-disjoint, they constitute a cycle, leading to a contradiction to the assumption that $v$ is a tree node.
    Otherwise, the two paths share at least one edge and hence intersect at some node. 
    Any such node lies on a cycle formed by the union of the two paths, and therefore must be a core node. 
    This contradicts the choice of the node $u$ as the core node nearest to $v$, since such a core node would be closer to $v$ than $u$.
    Thus, there exists exactly one $x \in N_2(v)$ with $dist(x, u) = k-1$. 
    This completes the proof of Claim \ref{claim:unique-parent}.   
    \end{proof}

        By the similar argument presented in Claim \ref{claim:unique-parent}, we can say that any node $w \in N_1(v)$ cannot have the same distance ($k$) from $u$ as $v$, i.e., $dist(w,u) \neq dist(v,u)$.
        Hence, any neighboring node of $v$ has either the distance $k-1$ or $k+1$ from $u$.
        Moreover, the similar argument also establishes that there is exactly one node at $k-2$ distance\footnote{With a little abuse of notation, if the node $v$ sees the base string $00110$, with the central bit being the label of $v$, and $dist(v,u) = 1$, by a node at $k-2$ distance away from $u$, we mean the node whose label is the first bit $0$. In such a case, this node is actually a core node, whose one of the neighbors is $u$.} away from $u$. 

        Let us consider the $v$-centered route $P= v_1v_2v_3v_4v_5$ with $v_3 = v$ and $v_1$ (resp. $v_2$) is the node at $k-2$ (resp. $k-1$) distance away from $u$. The nodes $v_4$ and $v_5$ are any two nodes such that $dist(v_4, u) = k+1$, $dist(v_5, u) = k+2$, and $v_5$ is a neighbor of $v_4$.

        Under the proof $L_O^*$, the nodes on $P$ must receive labels such that they constitute a base string $S_{bi}$ for $i\in \{1,2,\dots, 6\}$ (see Table \ref{table:refined-membership}).
        Now our concern would be: can there be any other base string $S_{bj}$, $i\neq j$, such that it points to another neighbor (may or may not be on $P$) of $v$ as the parent and \textsc{Discover-Parent}$(v)$ returns $w$?
        Notice from Table \ref{table:refined-membership} (see the third column), that for any two base strings $S_{bi}$ and $S_{bj}$ with $i\neq j$ and $i,j \in \{1, 2,\dots 6\}$, we cannot have $S_{bi} \in t(S_{bj})$ and $S_{bj} \in t(S_{bi})$ simultaneously.
        It is also clear that none of the base strings is a palindrome. 
        Hence, there is exactly one $v$-centered route that induces a base string such that any other length $5$ string corresponding to a $v$-centered route belongs to the three string set of the base string. \textsc{Discover-Parent}$(v)$ returns the parent whose label is the second bit of the base string.

        $\blacktriangleright$ \textbf{Case 2 ($v$ is a core node):} By construction of $L_O^*$, node $v$ lies on a cycle or on a path connecting two cycles.
        Hence, there exist two neighbors $u$ and $w (\neq u)$ of $v$ such that $L_O^*(v)=L_O^*(u)=L_O^*(w)=0$. Hence, it identifies itself as a core node, according to \textsc{Discover-Parent}($v$).
        
This completes the proof of Lemma \ref{lem:correctness-algorithm2}.
 \end{proof}

 \begin{lemma}
    \label{lem:2labels-3view-graph-with-cycles}
        
        If $G$ contains a cycle, every node accepts under $L_O^*$ according to the algorithm $\mathcal{A}_0$-\textsc{2Labels}.
    \end{lemma}

    \begin{proof}
        $L_O^*$ assigns the proof label $0$ to all core nodes. 
        When node $v$ detects itself as a core node (the correctness of such detection is by Lemma \ref{lem:correctness-algorithm2}) and simulates the same subroutine \textsc{Discover-Parent}() for every neighbor, two of its neighbors are found to be core nodes.
        Such a simulation for neighbors is possible because $v$ has a view distance of $3$. 
        If $v$ determines itself as a tree node (and thus returns its parent $p$ by Lemma \ref{lem:correctness-algorithm2}) and runs the same sub-routine for
        its parent $p$ to check for consistency (i.e., whether $p$ identifies $v$ as its child or not), $v$ determines that the parent of $p$ is not $v$ in case of $p$ being a tree node (else it is a core node). Such a check is possible because $ v$ has a $3$-hop view distance.
        Hence, all nodes accept when $G$ contains a cycle. 
    \end{proof}

 \begin{lemma}
 \label{lem:2labels-3view-graph-without-cycles}
     If $G$ does not contain a cycle, there is a node that rejects for any given proof $L$ according to the algorithm $\mathcal{A}_0$-\textsc{2Labels}.
 \end{lemma}

 \begin{proof}
    Assume for the sake of contradiction that all nodes accept under some proof $L$ on $G$ according to Algorithm $\mathcal{A}_0$-\textsc{2labels}.
    Now the question is: Can there be a node in $G$ that accepts by identifying itself as a core node under the proof $L$? Let us assume that such a core node exists in $G$.
    By definition, all core nodes form a single connected component in $G$.
    Observe that each core node must have at least two other core nodes as its neighbors.
    Therefore, the minimum degree of any node in the subgraph induced by the core nodes is at least $2$.
    However, since the graph $G$ is acyclic, every non-empty induced subgraph must contain a node of degree at most $1$.
    Therefore, an induced subgraph with minimum degree at least $2$ cannot exist inside $G$.
    This means that if $G$ is acyclic, there cannot be any core node. 

    Consequently, all nodes identify themselves as tree nodes under $L$.
    Let us pick one such tree node $v_0$. Since it accepts, it must have determined some node $v_1$ as its parent under the proof $L$. Applying the same reasoning repeatedly on node $v_1$ and subsequently on its parent $v_2$, we get a sequence of nodes $v_0, v_1, v_2, \dots$. 
    Since $G$ is finite and all tree nodes decide to accept under $L$, there exist indices $j < j'$ in the sequence such that $v_j = v_{j'}$.
    Hence, the subsequence $v_j, \dots v_{j'}$ forms a cycle, a contradiction to the fact that $G$ is acyclic.
    This completes the proof.
 \end{proof}

 By Lemma \ref{lem:2labels-3view-graph-with-cycles} and \ref{lem:2labels-3view-graph-without-cycles}, we prove completeness and soundness properties, and this ends the proof of Theorem \ref{thm:cycle_2labels_3viewdistance}.
 \end{proof}

\section{Algorithmic Framework for Proofs with Errors: \framework{} (\textbf{\texttt{r}}esilient \textbf{\texttt{e}}rror \textbf{\texttt{f}}ixing)}
    \label{subsec:framework}

We now consider the setting in which the adversary may introduce errors in the oracular proof $L_O$.
We assume that within the neighborhood $N_{2i+1}(v)$ for each node $v$, the adversarial proof $L_{adv}$ can differ from the oracular proof $L_O$ in at most $i$ nodes.

We want to design an algorithmic framework for verifying a graph property $\mathcal{P}$ that can tolerate the aforementioned adversarial modifications. In this process, we also assume that each node has a view distance of $2i+1$.
It is evident that we cannot restrict the view distance of the nodes to $1$ (like we often see it in the classical proof labeling schemes in the literature), as the adversary could introduce errors in the proof labels of all neighbors of a node (when the degree of the node is less than the number of errors $i$), and such a node can easily be misguided in taking the decision (accept or reject). 
We discuss more on this when we establish the lower bounds (or impossibilities) on the required view distance in Sec. \ref{sec:lower-bound-erroneous-labelings} to understand the relation between the number of errors $i$ and the view distance.

Let $\mathcal{P}$ be a graph property to be verified under the LCP model. 
Let $\mathcal{A}_0$ be the \emph{base verification algorithm} to verify $\mathcal{P}$ under an oracular proof $L_O$, when $i = 0$.
Consider Algorithm \ref{alg:framework}, which we refer to as \framework{} (abbreviation of \textbf{\texttt{r}}esilient \textbf{\texttt{e}}rror \textbf{\texttt{f}}ixing) Framework.
It takes as input the base verification algorithm $\mathcal{A}_0$ and $i$, and applies a generic correction-and-check procedure to tolerate errors introduced by the adversary, as follows.

\begin{algorithm2e}
    Check every possible proof $l_v$ on the nodes in $N_{2i+1}(v)$ that differs from $L_{adv}$ in at most $i$ nodes within $N_{2i+1}(v)$ 

    \If{there is a proof $l_v$ such that all the nodes within $N_{2i}(v)$ accept according to $\mathcal{A}_0$}
        {
            $v$ accepts
        }

    \Else
        {
            $v$ rejects            
        }

	\caption{\framework{}-Framework~($\mathcal{A}_0$, $i$)}
 \label{alg:framework}
\end{algorithm2e}

\noindent $\blacktriangleright$ \textbf{General Idea.} 
In order to describe the framework, let us first define the notion of \emph{imagined proof} from the perspective of the node $v$, which captures the notion that each $v$ looks for corrections to the proof that it sees on the subgraph induced by the nodes in $N_{2i+1}(v)$. We will use this notion often in the future sections.

\begin{definition}[Imagined Proof]
    \label{def:imagine}
        A proof $l_v$ is called an imagined proof for node $v$ on the nodes within $N_{2i+1}(v)$, if $l_v$ is the proof that causes $v$ to accept, i.e., $l_v$ differs from $L_{adv}$ in at most $i$ nodes within $N_{2i+1}(v)$ and under $l_v$, all nodes within $N_{2i}(v)$ accept according to the base algorithm $\mathcal{A}_0$.
        To refer to this, we sometimes say that $v$ imagines $l_v$.        
\end{definition}

Intuitively, the imagined proof for a node $v$ means that $v$ locally modifies the proof labels of at most $i$ nodes within its view distance of $2i+1$, so that if the adversary has introduced at most $i$ errors within its $N_{2i+1}(v)$, it can find a correction for the proof labels of the nodes within $N_{2i+1}(v)$ and check whether all nodes within $N_{2i}(v)$ accept according to the base algorithm (which works without any errors) $\mathcal{A}_0$.
If such an imagined proof on $N_{2i+1}(v)$ exists, $v$ accepts, and rejects otherwise. 
The check is limited to $N_{2i}(v)$ (rather than $N_{2i+1}(v)$) because modifying a proof label of some node $u$ in $v$'s local computation may affect the decisions of $u$ and its neighbors under $\mathcal{A}_0$, and nodes on the boundary $N_{2i+1}(v) \setminus N_{2i}(v)$ may have neighbors outside $v$'s view whose labels are unknown to $v$ (and hence it cannot be sure about their decision).


\noindent $\blacktriangleright$ \noindent\textbf{Analysis of the Framework.} We analyze the correctness of the framework. We separate the analysis into a problem (or property $\mathcal{P}$)-specific component and a problem-independent component.
Before explaining more about that, let us start the analysis with the following lemma, which proves the completeness property. 

\begin{lemma}
    \label{lemma:graph-with-property-prediction-model}
    When a graph $G$ has the property $\mathcal{P}$, \framework{} makes every node $v$ accept even under an adversarial proof $L_{adv}$ with $v$ having a view distance of $2i+1$, where $\mathcal{A}_0$ is the algorithm for verifying $\mathcal{P}$ under an oracular proof $L_O$ on $G$ with a view distance $1$ for every node and $L_{adv}$ differs from $L_O$ in at most $i$ nodes within $N_{2i+1}(v)$.
\end{lemma}

\begin{proof}
    This is quite straightforward.
    Consider an arbitrary node $v$. 
    It checks every possible proof $l_v$ on the subgraph induced by the nodes within $N_{2i+1}(v)$ that differs from $L_{adv}$ in at most $i$ nodes in $N_{2i+1}(v)$. 
    Therefore, $v$ checks the proof that is identical to $L_O$ on all nodes in $N_{2i+1}(v)$.
    Since $L_O$ is free from errors and all nodes in $G$ accept under $L_O$ according to the verification algorithm $\mathcal{A}_0$, $v$ accepts according to \framework{}.
\end{proof}

We now move to the soundness property. 
As mentioned at the start of the analysis, by the problem-specific component, we mean the following proposition, which is relevant in the proof of the soundness property (i.e., when a graph $G \notin \mathcal{P}$).

\begin{proposition}
\label{prop:diff_view_dist}
    Consider any two adjacent nodes $u,v$ such that they imagine (recall Definition \ref{def:imagine}) proofs $l_u$ and $l_v$ respectively.
    Then $l_u(u)=l_v(u)$ and $l_u(v)=l_v(v)$.
\end{proposition}

We assume Proposition~\ref{prop:diff_view_dist} and prove soundness below. This proposition will be proved in the next section (Sec. \ref{sec:erroneous-labeling}), individually for each verification problem: cycle existence, cycle-freeness and bipartiteness, as it requires the inherent properties of their LCPs. 
We emphasize that it does not matter in soundness whether $L_{adv}$ differs from $L_O$ in at most $i$ nodes within $N_{2i+1}(v)$ for every node $v$. 
Even if $L_{adv} = L_O$, some node must reject if $G \notin \mathcal{P}$.

\begin{lemma}
    \label{lemma:graph-without-property-prediction-model}
    When a graph $G$ does not have the property $\mathcal{P}$, for any proof $L_{adv}$ (possibly adversarial or not), \framework{} makes at least one node reject.
\end{lemma}

\begin{proof}

    Assume to the contrary that there is an adversarial proof $L_{adv}$ such that all nodes in $G$ accept according to \framework{}. We now construct a proof $L$ on $G$ using this assumption and show that all nodes in $G$ accept according to the base verification algorithm $\mathcal{A}_0$ (which works correctly when there are no errors), which leads us to a violation of the soundness property of the error-free version of the LCP.

    Since all the nodes accept, then each node $v$ imagines a proof $l_v$ on the subgraph induced by the nodes in $N_{2i+1}(v)$.
    We define a proof $L$ on $G$ such that $L(v)=l_v(v)$ for each node $v$, i.e., the proof $L$ assigns the imagined proof label of a node $v$ onto itself.
    
    Now we show that $\mathcal{A}_0$ makes all nodes accept under $L$. 
    Consider any node $v$ and an arbitrary neighbor $u$ of $v$. 
    Since $v$ accepts according to \framework{} with an imagined proof $l_v$, then it checks that all nodes, including $v$ itself in $N_{2i}(v)$, accept according to $\mathcal{A}_0$. Since we assumed that Proposition~\ref{prop:diff_view_dist} holds true, we have $l_v(u)=l_u(u)=L(u)$ for all $u \in N_1(v)$. 
    Thus, the proof $L$ on nodes in $N_1(v)$ is consistent with the proof $l_v$ and node $v$ accepts $L$ according to $\mathcal{A}_0$. Therefore, every node $v$ in $G$ accepts under the proof $L$ according to $\mathcal{A}_0$, a contradiction for a graph not satisfying $\mathcal{P}$.
    \end{proof}

Lemma~\ref{lemma:graph-with-property-prediction-model} (completeness) and~\ref{lemma:graph-without-property-prediction-model} (soundness) together establish the following theorem.

\begin{theorem}
\label{thm:correctness-framework}
Let $\mathcal{A}_0$ be the algorithm for verifying a graph property $\mathcal{P}$ under an oracular proof $L_O$ with a view distance $1$ for every node of a given graph $G$ and the proof $L_{adv}$ differ from $L_O$ in at most $i$ nodes within $N_{2i+1}(v)$.
The proof $L_{adv}$ together with \framework{} establishes an LCP for verifying $\mathcal{P}$ with every node having a view distance of $2i+1$, under the assumption that Proposition~\ref{prop:diff_view_dist} holds if $G \notin \mathcal{P}$.
\end{theorem}

\section{Verification Problems with Erroneous Proofs}
\label{sec:erroneous-labeling}

We prove Proposition \ref{prop:diff_view_dist} for three verification problems: cycle existence, cycle-freeness and bipartiteness.
Here we present the proof of Proposition \ref{prop:diff_view_dist} for cycle-freeness and bipartiteness. 
For cycle-freeness, the proof closely follows the proof for cycle existence. Therefore, to avoid repetition, we refer to the corresponding arguments from Sec. \ref{sec:erroneous-labeling} whenever applicable.

\subsection{Cycle Existence with Proofs with Errors}
\label{subsec:cycle-detection-with-errors}

\noindent \textbf{Setup (Oracular Proof and Base Algorithm):}
We denote the oracular proof for cycle existence by $L_O^1$, and the base (error-free variant) algorithm by $\mathcal{A}_0$-\textsc{cycle}.
Although Sec.~\ref{sec:3labels} presents a proof for cycle existence using only $3$ distinct proof labels and view distance $1$, for convenience we consider a proof $L_O^1$ with $n$ distinct proof labels. In $L_O^1$, each tree node is assigned its exact (BFS) distance to the nearest core node, rather than this distance modulo $3$. 
$\mathcal{A}_0$-\textsc{cycle} is adapted from $\mathcal{A}_0$ accordingly: a node with proof label $0$ accepts if it has at least two neighbors with proof label $0$ and no neighbor with proof label at least $2$, while a node $v$ with $L(v)>0$ accepts if it has exactly one neighbor with proof label $L(v)-1$ and all remaining neighbors with proof label $L(v)+1$.

By Theorem~\ref{thm:correctness-framework}, it remains to prove Proposition~\ref{prop:diff_view_dist} for cycle existence. Recall that in the soundness proof, we assume for contradiction that all nodes in an acyclic graph $G$ accept under some proof $L$.
We first state two key claims.

    \begin{remark}
        \label{remark:labels-differ-by-1}
        Under any proof $L$, the proof labels of any two adjacent nodes must differ by at most one, as otherwise $\mathcal{A}_0$-\textsc{cycle} would make at least one node reject.
    \end{remark}

    \begin{claim}
    \label{lem:path_labeling}
    
        Let $P = (v_0, v_1, \ldots, v_k)$ be a path in a graph $G$. Consider any fixed $a,b \in \mathbb{N}\cup \{0\}$.
        Then, there exists at most one proof $L$ on $P$ such that $L(v_0) = a$, $L(v_k) = b$, and every node on $P$ accepts according to $\mathcal{A}_0$-\textsc{cycle}.
    \end{claim}

\begin{proof}[Proof of Claim \ref{lem:path_labeling}]
        We prove the claim by a case distinction depending on whether the endpoint labels $a$ and $b$ are both zero.

        $\blacktriangleright$ \textbf{Case A ($a \neq 0$ or $b \neq 0$):} In this case, we use induction on the path length $k'$.

        \textit{Inductive hypothesis $H(k')$:} For any $a',b' \in \mathbb{N}\cup \{0\}$ with $a' \neq 0$ or $b' \neq 0$, there exists at most one proof $L$ on a path of length $k'$ with endpoint nodes assigned the proof labels $a'$ and $b'$ under which every node accepts according to $\mathcal{A}_0$-\textsc{cycle}.

        \textit{Base case $H(1)$:} A path $(u_0, u_1)$ with fixed endpoint labels $L(u_0) = a'$ and $L(u_1) = b'$ either satisfies the acceptance conditions of $\mathcal{A}_0$-\textsc{cycle} or not. In either case, at most one proof exists, so $H(1)$ holds.

        \textit{Inductive step ($H(k') \Rightarrow H(k'+1)$):}
        Let $P'=(u_0,u_1,\dots,u_{k'+1})$ be a path of length $k'+1$ with endpoint labels $a',b'$ where $a' \neq 0$ or $b' \neq 0$. 
        If no proof satisfies these constraints while making all nodes accept under $\mathcal{A}_0$-\textsc{cycle}, then $H(k'+1)$ holds trivially. 
        
        Otherwise, let $L'$ be such a proof on $P'$ and $m = \max\{L'(u_0),\dots,L'(u_{k'+1})\}$.
        Since the proof label of at least one endpoint node is non-zero, $m > 0$.
        We show that no internal node $u_j$ ($1 \leq j \leq k'$) can have label $m$.
        Indeed, if $L'(u_j) = m > 0$, then by the acceptance rule of $\mathcal{A}_0$-\textsc{cycle}, $u_j$ requires exactly one neighbor with label $m-1$ and all others with label $m+1$. But both neighbors $u_{j-1}$ and $u_{j+1}$ satisfy $L'(u_{j-1}) \leq m$ and $L'(u_{j+1}) \leq m$ by maximality, so neither has label $m+1$. If either has label $m$, the acceptance condition is also violated. Hence $u_j$ rejects---a contradiction. Therefore, no internal node of $P'$ can have label $m$.

        So $m$ appears only at endpoints of $P'$. Without loss of generality, $L'(u_0)=m=a'$.
        The acceptance rule forces $L'(u_1) \in \{m-1, m+1\}$, and by maximality $L'(u_1) = m-1 = a'-1$.
        By $H(k')$ applied to the subpath $(u_1, \dots, u_{k'+1})$ with endpoint labels $a'-1$ and $b'$, at most one proof of this subpath exists.
        (If $a'-1 = 0$ and $b' = 0$, we apply Case~B below to this subpath, which also gives uniqueness.)
        Hence, at most one proof on $P'$ exists, establishing $H(k'+1)$.

        $\blacktriangleright$ \textbf{Case B ($a = b = 0$):}
        Let $L$ be a proof on the path $P = (v_0, v_1, \dots v_k)$ such that all nodes accept according to $\mathcal{A}_0$-\textsc{cycle}, and  $m$ be the maximal proof label on this path. 
        We claim that $m = 0$.
        If $m > 0$, then any node $v$ with the proof label $m$ must either be adjacent to another node with proof label $m$ or have two neighbors with label $m-1$.
        In either case, $v$ rejects according to $\mathcal{A}_0$-\textsc{cycle}, contradicting the assumption that all nodes accept.
         Hence $m = 0$, and the unique accepted proof is $L(v_j)=0$ for all $j$.

        Together, Cases~A and~B establish Claim \ref{lem:path_labeling}.
    \end{proof}

We state one more claim before proving Proposition~\ref{prop:diff_view_dist}.
Consider a path $P$ in the acyclic graph $G$ with at least $4i+2$ nodes and select two neighboring central nodes $u$ and $v$ on $P$.
For the rest of the analysis, we define the following terminology exclusive to these proofs. 

$\blacktriangleright$ We refer to a path $P_l = (v_1, v_2, \dots, v_k)$ as a \emph{long path} if $k \geq 4i+2$. 

$\blacktriangleright$ We refer to a path $P_s = (v_1, v_2, \dots, v_k)$ as a \emph{short path} if $2i+2 \leq k \leq 4i+1$.

$\blacktriangleright$ On any (long or short) path $P = (v_1, \dots, v_k)$ with $u = v_{2i+1}$ and $v = v_{2i+2}$, the sub-path $(v_1, v_2, \dots, v_{2i+1}=u)$ as $\texttt{left}(P)$ and the sub-path $(v = v_{2i+2}, \dots v_k)$ as $\texttt{right}(P')$. 
The notion of \texttt{left} or \texttt{right} is only for descriptive purposes. 

$\blacktriangleright$ The subgraph $T_{u,v}$ induced by the nodes in $N_{2i+1}(u) \cap N_{2i+1}(v)$ is a tree, since $G$ is acyclic. We refer to it as the \emph{surrounding tree of $u$ and $v$}. 
A node $w \in T_{u,v}$ is \emph{outer} if it is adjacent to some node outside $T_{u,v}$, and \emph{inner} otherwise.

$\blacktriangleright$ For each node $u$ and $v$, the nodes in the surrounding tree $T_{u,v}$ are assigned some proof labels in their local computation according to the imagined proof (recall Definition \ref{def:imagine}) $l_u$ and $l_v$.
We refer to the node in $T_{u,v}$ with the smallest label according to the imagined proof $l_u$ (resp. $l_v$) as $u_{root}$ (resp. $v_{root}$).

    \begin{claim}
    \label{claim:diff_view_dist}
        Consider two adjacent nodes $u$ and $v$ with their respective imagined proofs (ref. Definition \ref{def:imagine}) $l_u$ and $l_v$. If there exists a long path $P_l = (v_1,v_2,\dots,v_{4i+2})$ such that $u=v_{2i+1}$ and $v=v_{2i+2}$, then $l_u(u)=l_v(u)$ and $l_u(v)=l_v(v)$, i.e., their imagined proofs agree on each other's proof labels.
    \end{claim}
\begin{proof}[Proof of Claim \ref{claim:diff_view_dist}]        
        
        Notice that all the nodes on $P_l$ are within the view distance of both $u$ and $v$.
        Each imagined proof $l_u$ (resp.\ $l_v$) can differ from $L_{adv}$ in at most $i$ nodes.
        Since $\texttt{left}(P_l)$ and $\texttt{right}(P_l)$ each contain $2i+1$ nodes and the total number of disagreements between $\{l_u, l_v\}$ and $L_{adv}$ is at most $2i$, there exists at least one node $v_l \in \texttt{left}(P_l)$ and at least one node $v_r \in \texttt{right}(P_l)$ such that $l_u(v_l) = l_v(v_l)=L_{adv}(v_l)$ and $l_u(v_r)=l_v(v_r)=L_{adv}(v_r)$, i.e., both $u$ and $v$ agree on the proof labels of the nodes $v_l$ and $v_r$. 
        Let $\hat{P}$ be the sub-path of $P_l$ from $v_l$ to $v_r$ (well-defined since $G$ is acyclic). 
        Both $l_u$ and $l_v$, restricted to $\hat{P}$, are proofs with the same endpoint labels under which all nodes accept according to $\mathcal{A}_0$-\textsc{cycle}. 
        By Claim~\ref{lem:path_labeling}, there exists at most one proof on this $\hat{P}$. 
        As a consequence, $u$ and $v$ agree on the proof labels of all the nodes on $\hat{P}$, and in particular $l_u(u)=l_v(u)$ and $l_u(v)=l_v(v)$, as both $u,v \in \hat{P}$.
    \end{proof}

    We now prove Proposition~\ref{prop:diff_view_dist} for cycle existence.




\begin{proof}[Proof of Proposition~\ref{prop:diff_view_dist} for Cycle Existence]
If there is no long path in $G$, it is small enough in the sense that there is a node seeing the whole $G$ within its view distance.
Such a node simply rejects, contradicting our assumption in the proof of soundness that all nodes accept.

So we focus on a more intricate case, when $G$ has at least one long path (i.e., a path with at least $4i+2$ nodes). 
Here, our interest lies in two types of adjacent nodes $u$ and $v$. 
First is where both $u$ and $v$ lie on a long path $P_l$.
In such a scenario, the proposition directly follows from Claim~\ref{claim:diff_view_dist}.
The second type of our interest is where one of $u$ and $v$ lies on a short path $P_s = (v_1, \dots, v_k)$ with $2i+2 \leq k \leq 4i+1$.

    Without loss of generality, let $v$ be such a node, i.e., $u = v_{2i+1}$ and $v = v_{2i+2}$. 
    Clearly, $\texttt{right}(P_s)$ has at most $2i$ nodes.
    Also, notice that every node in $\texttt{right}(P_s)$ is within $N_{2i}(v)$ and hence in $v$'s local computation, they all accept under the imagined proof $l_v$ on $\texttt{right}(P_s)$ according to $\mathcal{A}_0$-\textsc{cycle}.
    Furthermore, all the nodes in $\texttt{right}(P_s)$ lie within the view distance of $u$, i.e., within $N_{2i+1}(u)$.
       Below, we state some simple yet useful observations as remarks.

    \begin{remark}
        \label{remark:root-cannot-lie-inside-viewdistance}
        For node $w \in \{u,v\}$ with a view-distance $2i+1$, if the root $w_{root}$ is an inner node of $T_{u,v}$, then $l_w(w_{root})$ cannot be non-zero, as otherwise $w_{root}$ has no parent under the proof label assignment of the imagined proof $l_w$, and $w$ cannot accept in that case.
    \end{remark}

    \begin{remark}
        \label{remark:series-of-zeros}

        For a node $w \in \{u,v\}$, if $l_w(w_{root}) = 0$ and $w_{root}$ is an inner node of $T_{u,v}$, then there must exist at least one path of nodes all with proof labels $0$ between two outer nodes of $T_{u,v}$ under the imagined proof $l_w$ of $w$ that contains $w_{root}$.
    \end{remark}

    \begin{remark}
        \label{remark:root-cannot-lie-in-right}
        For a short paths $P_s = (v_1, \dots, v_k)$, $2i+2 \leq k \leq 4i+1$ with $u = v_{2i+1}$ and $v = v_{2i+2}$, we have $v_{root} \notin \bigcup_{P_s}\texttt{right}(P_s)$.
    \end{remark}

    As a consequence of Remark \ref{remark:root-cannot-lie-in-right}, we say that, in the imagined proof $l_v$ of $v$, any two consecutive nodes $v_r$ and $v_{r+1}$ in $\bigcup_{P_s}\texttt{right}(P_s)$ satisfy $l_v(v_r) < l_v(v_{r+1})$.
    If both $u$ and $v$ agree on the root of $T_{u,v}$, we are done: the agreed root must lie on $\bigcup_{P_s} \texttt{left}(P_s)$ and on an arbitrary short path $P_s$ passing through the root, and the node $u$ and $v$, there cannot be two monotonically increasing sequence of proof labels.
    When $u$ and $v$ disagree on the root, we prove the following two claims, depending on whether the proof labels of $u_{root}$ or $v_{root}$ are $0$.

    \begin{claim}
        \label{claim:one-root-nonzero}
        If at least one of $l_u(u_{root})$ and $l_v(v_{root})$ is not $0$, Proposition \ref{prop:diff_view_dist} holds.
    \end{claim}

    \begin{proof}
        We already know from Remark \ref{remark:root-cannot-lie-in-right} that $v_{root}$ cannot lie on any $\texttt{right}(P_s)$ for any $P_s$ passing through $u$ and $v$.
        Therefore, $v_{root} \in \bigcup_{P_s} \texttt{left}(P_s)$.
        Let us first assume that $u_{root} \in \bigcup_{P_s} \texttt{left}(P_s)$.
        Consider the following four paths: ($P_1$) from $u$ to $u_{root}$, ($P_2$) from $v$ to $v_{root}$, ($P_3) = (P_1 \cap P_2) \cup \texttt{right}(P')$ and ($P_4$) from $u_{root}$ to $v_{root}$.
        By construction, $P_3$ and $P_4$ have exactly one node in common, say $v_{k'}$. 
        For any two consecutive nodes $v_j$ and $v_{j+1}$ on $P_3$, both $u$ and $v$ satisfy $l_u(v_j) < l_u(v_{j+1})$ and $l_v(v_j) < l_v(v_{j+1})$, as the distance from the root to $v_{j+1}$ (and hence the proof label) must be one larger than that of $v_j$ from any of $u$ and $v$'s perspective.
        Observe that for each node $x \in P_3$, we have $l_u(x) = l_u(v_{k'}) + dist(x,v_{k'})$ and $l_v(x) = l_v(v_{k'}) + dist(x,v_{k'})$.
        If $l_u(v_{k'}) = l_v(v_{k'})$, from the above two identities, we are done proving Proposition \ref{prop:diff_view_dist}, as we have $l_u(u) = l_v(u)$ and $l_v(u) = l_v(v)$ in particular.
        
        Hence, let us continue assuming that $l_u(x) \neq l_v(x)$ for any $x \in P_3$.
        Notice that on the path $P_4$, $u$ imagines in $l_u$ that the proof labels monotonically decrease starting from $v_{root}$ toward $u_{root}$, whereas $v$ imagines the opposite (i.e., monotonically increase) in $l_v$ (because of disagreement on the location of the root of $T_{u,v}$).
        Consequently, there could be at most one node $v_{k''}$ on this path $P_4$ such that $l_u(v_{k''}) = l_v(v_{k''})$.
        Since $l_u(x) \neq l_v(x)$ for any $x \in P_3$ and there is at most one agreement on $P_4$, we can say that there is at most one agreement on $P_3 \cup P_4$. 
        Evidently, $P_2 \subseteq P_3 \cup P_4$ and $P_2$ contains $2i+2$ nodes.
        Therefore, there are at least $2i+1$ disagreement on the path $P_2$, which is impossible, since $u$ and $v$ can have at most $2i$ disagreement in total. 

        Now, if $u_{root} \in \bigcup_{P_s} \texttt{right}(P_s)$, we construct a path $P_5$ with at least $2i+2$ nodes as follows.
        When both $l_u(u_{root}), l_v(v_{root}) \neq 0$, by Remark \ref{remark:root-cannot-lie-inside-viewdistance}, both $u_{root}$ and $v_{root}$ are inner nodes of $T_{u,v}$ and hence we let $P_5$ be the path from $u_{root}$ to $v_{root}$.
        However, if one of $l_u(u_{root}), l_v(v_{root})$ is $0$ and the other is not, by Remark \ref{remark:series-of-zeros}, there exists at least one path of $0$-labeled nodes between two outer nodes of $T_{u,v}$ that contains the root with label $0$.
        Without loss of generality, if $l_u(u_{root})= 0$, then there is an outer node $x \in \bigcup_{P_s}\texttt{right}(P_s)$ with $l_u(x) = 0$; we let $P_5$ be the path connecting $v_{root}$ through $u_{root}$ to $x$.
        The disagreement between $u$ and $v$ on the root of $T_{u,v}$ implies that the child$\rightarrow$parent direction on $P_5$ is opposite in $l_u$ and $l_v$.
        Since $P_5$ has at least $2i+2$ nodes and the opposite orientations allow at most one agreement, there are at least $2i+1$ disagreements, leading to a contradiction as they can have at most $2i$ disagreement in total.

        Hence, $u$ and $v$ must agree on the root of $T_{u,v}$. So, on any path passing through $u$, $v$, and the root, $l_u$ and $l_v$ must have the same proof labels assigned, giving Proposition \ref{prop:diff_view_dist}.
    \end{proof}

    \begin{claim}
        \label{claim:both-root-zero}
        If both $l_u(u_{root})$ and $l_v(v_{root})$ are $0$, Proposition \ref{prop:diff_view_dist} holds.
    \end{claim}

    \begin{proof}
        Let $C_u$ (resp.\ $C_v$) denote the maximal connected set of nodes in $T_{u,v}$ assigned the proof label $0$ under $l_u$ (resp.\ $l_v$).
        Each of $C_u$ and $C_v$ is a single connected component, since two disjoint $0$-components in $T_{u,v}$ would force an intermediate non-zero-labeled node to reject under $\mathcal{A}_0$-\textsc{cycle}.
        We distinguish two cases.

        \textbf{Case $C_u \cap C_v = \emptyset$:} By Remark~\ref{remark:series-of-zeros}, there exist outer nodes $v_{k'} \in C_u$ and $v_{k''} \in C_v$ with $v_{k'} \neq v_{k''}$.
        Let $P_4$ be the path from $v_{k'}$ to $v_{k''}$, and define $P_3$ as in Claim~\ref{claim:one-root-nonzero}.
        The same counting argument applies: along $P_4$, the proof labels in $l_u$ and $l_v$ increase in opposite directions from their respective $0$-components, allowing at most one agreement, while on $P_3$ there are no agreements.
        Using the same argument as Claim \ref{claim:one-root-nonzero} on the number of total disagreements between $u$ and $v$, we get Proposition \ref{prop:diff_view_dist}.

        \textbf{Case $C_u \cap C_v \neq \emptyset$:}
        Let $x \in C_u$ and $y \in C_v$ be the nodes in their respective components closest to $u$.
        If $x = y$, then $l_u$ and $l_v$ share a common $0$-labeled node. Since labels strictly increase by $1$ along any simple path moving away from $x$ (by the acceptance rule of $\mathcal{A}_0$-\textsc{cycle}), Claim~\ref{lem:path_labeling} implies that $l_u$ and $l_v$ agree on every node reachable from $x$ within $T_{u,v}$; in particular $l_u(u)=l_v(u)$ and $l_u(v)=l_v(v)$.
        If $x \neq y$, assume without loss of generality that $dist(x, u) \leq dist(y, u)$.
         Consider a path that contains $\texttt{right}(P')$ (for an arbitrary short path $P'$), goes through $u, v$ and $x$ and reaches an outer node in $C_u$ without entering the component $C_v$ beyond their intersection. Let us take this path as $P_4$.
         Then, $l_u(w) \neq l_v(w)$ for each node $w \in P_4$. Notice that the path $P_4$ contains at least $2i+2$ nodes (as it contains 
         a sub-path from $v$ to the outer node).
         Therefore, there are at least $2i+2$ disagreements, which is impossible, since $u$ and $v$ can have at most $2i$ disagreements.
         We have analyzed all the cases, so the claim holds.
             \end{proof}
              Claim \ref{claim:one-root-nonzero} and \ref{claim:both-root-zero} together establish Proposition \ref{prop:diff_view_dist}.
\end{proof}

\subsection{Cycle-freeness with Proofs with Errors}
\label{subsec:cycle-absence-erroneous-labels}

As \framework{} relies on the base verification algorithm $\mathcal{A}_0$ for making decisions, we briefly recall the oracular proof and the corresponding error-free verification algorithm for cycle-freeness, both of which are already known in the literature~\cite{ostrovsky2017space,goos2016locally}.

\vspace{2mm}
\noindent \textbf{Setup (Oracular Proof and Base Algorithm):} We denote the oracular proof for cycle-freeness by $L_O^2$, and the base (error-free variant) algorithm by $\mathcal{A}_0$-\textsc{cycle-free}.
Under the oracular proof $L_O$, a designated root receives the proof label $0$, and all other nodes receive their exact BFS distance to this root. 
The base verification algorithm $\mathcal{A}_0$-\textsc{cycle-free} checks that any node with a proof label $> 0$ has exactly one neighbor with a proof label smaller by $1$ (its parent) and all other neighbors with proof labels larger by $1$ (its children).

We call \framework{}-Framework$(\mathcal{A}_0\textsc{-cycle-free}, i)$, and to establish its correctness with respect to cycle-freeness, we must show that Proposition \ref{prop:diff_view_dist} holds. 




\begin{proof}[Proof of Proposition \ref{prop:diff_view_dist} (for Cycle-freeness)]
    Recall that we started the soundness proof of \framework{} with the assumption for the sake of contradiction that $G$ contains a cycle, but all nodes accept under some adversarial proof $L_{adv}$.
    Consider any two adjacent nodes $u$ and $v$ that imagine (ref. Definition \ref{def:imagine}) valid proofs $l_u$ and $l_v$, respectively.
    
    Because $u$ accepts under the imagined proof $l_u$, all nodes in $N_{2i}(u)$ accept according to $\mathcal{A}_0$-\textsc{cycle-free}.
    This strict parent-child orientation enforced by $\mathcal{A}_0$-\textsc{cycle-free} implies that the subgraph induced by $N_{2i}(u)$ cannot contain a cycle; otherwise, at least one node on the cycle would have either two parents or no parent, causing $\mathcal{A}_0$-\textsc{cycle-free} to reject.
    Thus $N_{2i}(u)$ and $N_{2i}(v)$ are strictly acyclic.
    Consequently, the subgraph $T_{u,v} = N_{2i+1}(u) \cap N_{2i+1}(v)$, as defined in Sec. \ref{subsec:cycle-detection-with-errors}, is a tree. 

    If $u$ and $v$ lie on a long path (i.e., of the length at least $4i+2$) $P_l = (v_1, v_2, \dots v_k)$ with $k \geq 4i+2$ and $u = v_{2i+1}$, $v = 2i+2$, the conculsion follows directly from the same argument presented in Claim \ref{claim:diff_view_dist} for cycle existence: the proof labels on the long path are unique determined, so the two imagined proofs must agree on the proof labels of $u$ and $v$.
    Thus it remains to consider the case where $u$ and $v$ lie on a short path $P_s = (v_1, v_2, \dots v_k)$ with $2i+2 \leq k \leq 4i+1$ and $u = v_{2i+1}$, $v = 2i+2$.

    Under the cycle-freeness verification algorithm $\mathcal{A}_0$-\textsc{cycle-free}, every node with a positive proof label must have exactly one neighbor whose proof label is smaller by one, while all its other neighbors must have proof labels larger by one, and moreover the node with the proof label $0$ is identified as the root. 
    Consequently, once the root of $T_{u,v}$ is fixed, the proof labels on each simple path in $T_{u,v}$ are uniquely determined by distance from that root.
    Similar to the cycle existence proof in Sec. \ref{subsec:cycle-detection-with-errors}, let $u_{root}$ and $v_{root}$ be the roots (i.e., the node with the smallest proof label) of $T_{u,v}$ under the imagined proof $l_u$ and $l_v$, respectively.

    If $u_{root} = v_{root}$, then both imagined proofs induce the same parent-child orientation on the tree $T_{u,v}$.
    Since both $l_u$ and $l_v$ use the same root and every node in $T_{u,v}$ must accept according to $\mathcal{A}_0$-\textsc{cycle-free}, the proof labels assigned by the two imagined proofs must coincide on all nodes of $T_{u,v}$. In particular, $l_u(u) = l_v(u)$ and $l_u(v) = l_v(v)$, and we are done in this case.

    Let us suppose that $u_{root} \neq v_{root}$. 
    Each of $l_u$ and $l_v$ may differ from $L_{adv}$ in at most $i$ nodes inside $N_{2i+1}(u)$ and $N_{2i+1}(v)$, respectively.
    Thus, inside $T_{u,v}$, the two imagined proofs together can disagree with $L_{adv}$ on at most $2i$ nodes.
    Since $T_{u,v}$ is a tree, $T_{u,v}$ must possess \emph{outer nodes} (defined similar to Sec. \ref{subsec:cycle-detection-with-errors}), i.e., nodes adjacent to nodes outside $T_{u,v}$.
    Consider the simple path $P$ in $T_{u,v}$ connecting two outer nodes that passes through both $u_{root}$ and $v_{root}$. 
    We consider the following two cases, depending on whether the path $P$ passes through $u$ and $v$.

    \begin{itemize}
        \item Let $u, v \in P$. 
        In this case, the path $P$ must have at least $2i+2$ nodes, as it is a simple path in the (sub)tree $T_{u,v}$.
        Along $P$, the proof labels induced by $l_u$ are the distances from $u_{root}$, while the proof labels induced by $l_v$ are the distance from $v_{root}$.
        Hence the two proofs are monotone in opposite directions on $P$, and therefore they can agree on at most one node on $P$.
        This implies that there are more than $2i$ disagreement in total, which contradicts the fact that the two imagined proofs together can disagree with $L_{adv}$ on at most $2i$ nodes.

        \item Let $u, v \notin P$. In this case, there is no guarantee that the above path $P$ contains at least $2i+2$ nodes. 
        Hence, we consider a node $w \in P$ that is closest to node $v$ (and thus to $u$ too). Define $C$ as the set of all nodes in $T_{u,v}$ reachable from $w$ without traversing any other node in $P$. 
        Since $u, v \notin P$, the adjacent nodes $u$ and $v$ both belong to the subtree induced by $C$.
        Because both $u_{root}$ and $v_{root}$ lie on $P$, the shortest path from any node $x$ to either root must pass through $w$ (you may think of $w$ as a gateway node in $P$). Consequently, under the base verification algorithm $\mathcal{A}_0$-\textsc{cycle-free}, the imagined proof labels of all nodes in $C$ are deterministically tied to $w$ and must strictly increase as their distance from $w$ increases. In particular, for any $x \in C$, $l_u(x) = l_u(w) + dist(w,x)$ and $l_v(x) = l_v(w) + dist(w,x)$. 
        If $l_u(w) = l_v(w)$, both $l_u$ and $l_v$ agree on every node in $C$ and hence on $u$ and $v$ too, so the proposition holds.
        On the other hand, if $l_u(w) \neq l_v(w)$, $l_u$ and $l_v$ are forced to disagree on every node in $C$. In this case, take a path $P'$ inside $T_{u,v}$ connecting two outer nodes of $T_{u,v}$ that passes through $u_{root}, w, u$ and $v$ (we can also take $v_{root}$ in place of $u_{root}$). 
        By construction, $P'$ contains at least $2i+2$ nodes. Since there could be at most one node on $P$ where $l_u$ and $l_v$ can agree on, the two imagined proofs disagree on more than $2i$ nodes on $P'$, leading to a contradiction. 
        
    \end{itemize}

    Hence, as a conclusion, we can say that it is not possible to have $u_{root} \neq v_{root}$.
    This proves Proposition \ref{prop:diff_view_dist}.

\end{proof}

\subsection{Bipartiteness with Proofs with Errors}
\label{subsec:bipartiteness-erroneous-labeling}

Before moving to the proof of Proposition \ref{prop:diff_view_dist}, let us mention the trivial (oracular) LCP for bipartiteness.

\vspace{2mm}
\noindent \textbf{Setup (Oracular Proof and Base Algorithm):} In a bipartite graph $G = (X,Y)$, we can simply use $2$ distinct proof labels and view distance $1$ for each node, where each node in $X$ gets $0$, and each node in $Y$ gets $1$. 
For verification, we denote the base (error-free variant) algorithm running at each node by $\mathcal{A}_0$-\textsc{bipartite}.
It makes a node $v$ accept if all $u \in N_1(v)$ have a different proof label than its own.

We call \framework{}-Framework$(\mathcal{A}_0\textsc{-bipartite}, i)$, and to establish its correctness with respect to bipartiteness, we must show that Proposition \ref{prop:diff_view_dist} holds.
The proof of Proposition~\ref{prop:diff_view_dist} is considerably simpler than in the case of cycle existence or cycle-freeness, and therefore we present only a concise argument. The proposition for bipartiteness can be rewritten as follows.




\begin{proof}[Proof of Proposition \ref{prop:diff_view_dist} (for Bipartiteness)]
We do not require the notions defined for the proof of Proposition \ref{prop:diff_view_dist} in the case of cycle existence or cycle-freeness (e.g., long path, short path).
Consider a simple path $P = (v_1, \dots, v_k)$ with $2i+2 \le k \le 4i+2$ nodes such that $u = v_{2i+1}$ and $v = v_{2i+2}$. 
In particular, $P$ contains strictly more than $2i$ nodes.
On this path $P$, the nodes $u$ and $v$ together can disagree $L_{adv}$ on the proof labels of at most $2i$ nodes. 
Since $|P| > 2i$, there exists at least one node $w \in P$ such that $l_u(w) = l_v(w)$.

Fix such a node $w$, and consider the sub-path of $P$ that connects $w$ to $u$ and $v$. Since every node accepts according to the verification algorithm $\mathcal{A}_0$\textsc{-bipartite}, each node on $P$ must have a proof label different from that of its neighbors. Hence, once the proof label of $w$ is fixed, the proof labels of all other nodes on $P$ are uniquely determined by the parity of their distance from $w$. As $l_u(w)=l_v(w)$, both imagined proofs induce the same proof labels on $P$. Consequently, $l_u$ and $l_v$ agree on every node of $P$, and in particular on $u$ and $v$.

\end{proof}

\section{Impossibility Results for Proofs with Errors}
\label{sec:lower-bound-erroneous-labelings}

We now investigate how the parameters $i$ (the number of errors) and the view distance relate.
We prove the following impossibility for cycle existence, which can even be extended to cycle-freeness and bipartiteness.

\begin{theorem}
    \label{thm:impos_cycle}
 If some $L_{adv}$ differs from $L_O$ in at most $i$ nodes even in the entire graph $G$, and even with no restriction on the number of distinct proof labels, then there exists no verification algorithm for cycle existence with a view distance of at most $i$ for every node.
\end{theorem}

\begin{proof}[Proof of Theorem \ref{thm:impos_cycle}]
    Assume, for the sake of contradiction, that such an algorithm $\mathcal{A}$ exists.
    Consider a graph $G_1$ that is a simple path of $3i+3$ nodes, denoted by $(u_0, u_1,u_2,\dots, u_{3i+2})$.
    Let $G_2$ be a graph consisting of a path $(v_0, v_1, v_2, \cdots, v_{3i+2})$ together with an additional triangle formed by three nodes $x, y$ and $z$, where node $x$ is connected with $v_{3i+2}$.
    We show that if all nodes in $G_2$ accept according to $\mathcal{A}$, then so do all nodes in $G_1$.

    Let us focus on the sub-path $S=(v_{i+1},\dots,v_{2i+1})$ of $i+1$ nodes in $G_2$. 
    Every node in $S$ is at a distance of at least $i$ from both the leaf $v_0$ and the triangle. 
    Hence, within the view distance $i$, each node in $S$ sees $2i+1$ nodes (including itself), all of which have degree $2$.
    For convenience, let us denote $l_j = L_O(v_{i+j})$ for $j = 1, 2, \dots, i+1$.
    Thus $(l_1, \dots, l_{i+1})$ is the sequence of proof labels on $S$ under the proof $L_O$. We now prove the following lemma.

    \begin{lemma}
    \label{lem:impos_acc_2i+1}
        Any node in $S$ that sees $2i+1$ nodes with degree $2$ each in a path and sees a consecutive sequence of proof labels $(l_1,\dots,l_{i+1})$ must accept.
    \end{lemma}

    \begin{proof}
        Fix a node $v \in S$ and assume that it has a view distance $i$.
        Since there could be at most $i$ errors in $G$, the proof labels of at least $i+1$ nodes must remain unchanged within $N_i(v)$.
        In particular, the adversary may arbitrarily modify the proof labels of any $i$ nodes within $N_i(v)$, while leaving the sequence of proof labels $(l_1,\dots,l_{i+1})$ unchanged.
        Since $v$ must accept every adversarial proof $L_{adv}$ differing from $L_O$ in at most $i$ nodes, its decision cannot depend on those $i$ potentially corrupted proof labels. Hence acceptance of $v$ according to the algorithm $\mathcal{A}$ is determined solely by the sequence of $i+1$ labels on $S$.
    \end{proof}


    Let us now construct a proof $L$ for $G_1$ such that $L(u_{i+j})=l_j$ for 
    $1 \leq j \leq i+1$. By Lemma~\ref{lem:impos_acc_2i+1}, all nodes $u_{i+1}, \dots, u_{2i+1}$ accept, as they 
    see a sequence of proof labels $(l_1,\dots,l_{i+1})$.
    To complete the construction of the proof $L$ for $G_1$, we must assign proof labels to the remaining nodes in $\{u_0,\dots,u_{i}\} \bigcup \{u_{2i+2},\dots,u_{3i+2}\}$.
    We do this by mapping the views of the nodes $v_0, \dots, v_i$ in $G_2$ onto both ends of $G_1$, as follows: for $0 \leq j \leq i$, we set $L(u_j) = L_O(v_j)$, and $L(v_{3i+2-j}) = L_O(v_j)$.
    We now evaluate the decision of the nodes in $G_1$ under $\mathcal{A}$.

    For any node $u_m$ with $0 \leq m \leq i$, its view extends at most to $u_{2i}$.
    By our construction and the definition of $S$, $L(u_j) = L_O(v_j)$ for $0 \leq j \leq 2i+1$.
    Thus, the local view of $u_m$ is identical to that of $v_m$ in $G_2$ under $L_O$.
    Because $\mathcal{A}$ makes $v_m$ accept even when there are $i$ nodes within $N_i(v_m)$ differing in proof labels from $L_O$, $\mathcal{A}$ must make $u_m$ accept.

    Finally, consider any node $u_m$ with $2i+2 \leq m \leq 3i+2$.
    Its local view contains nodes from the set $\{u_{2i+2}, \dots, u_{3i+2}\}$ along with at most $i$ nodes from $S$. 
    Due to our construction of $L$, we can compare the view of $u_m$ in $G_1$ and $v_{3i+2-m}$ in $G_2$.
    Both nodes 
    see an identical sequence of proof labels for the nodes 
    $\{u_{3i+2}, \dots, u_{2i+2}\}$ as $\{v_0, \dots v_i\}$.
    Since $u_m$ can see at most $i$ nodes from $S$, the proof labels of the nodes within $N_{i}(u_m)$ differ from the accepted oracular proof on $N_i(v_{3i+2-m})$ in at most $i$ nodes.
    Since $\mathcal{A}$ must tolerate $i$ adversarial modifications within the $i$-hop neighborhood for any node, it must make $u_m$ accept.

    Therefore, every node in $G_1$ accepts according to $\mathcal{A}$, even after it is acyclic, which is a contradiction to the soundness property. This completes the proof of Theorem \ref{thm:impos_cycle}.
\end{proof}

\begin{proof}[$\blacktriangleright$ Impossibility Proof of Theorem \ref{thm:impos_cycle} for Cycle-freeness:]
Suppose, for the sake of contradiction, that such an algorithm $\mathcal{A}$ exists for verifying cycle-freeness. 

Consider a graph $G_1 = (u_0, u_1, \dots, u_{3i+2})$ that is a simple path of $3i+3$ nodes.
Let us focus on the sub-path $S = (u_{i+1}, \dots, u_{2i+1})$ of $i+1$ nodes in $G_1$.
Every node in $S$ is at distance at least $i+1$ from both endpoints $u_0$ and $u_{3i+2}$. Hence, within the view distance $i$, each node in $S$ sees $2i+1$ nodes (including itself), all of which have degree $2$.
For convenience, let us denote $l_j = L_O(u_{i+j})$ for $j = 1, 2, \dots, i+1$.
Thus $(l_1, \dots, l_{i+1})$ is a sequence of proof labels on $S$ under $L_O$.
By the same argument as Lemma \ref{lem:impos_acc_2i+1}, the acceptance of each node $u_{i+j}$ for $1 \leq j \leq i+1$ is determined solely by the sequence $(l_1, \dots, l_{i+1})$ and its position within the view; the remaining $i$ proof labels in the view may take arbitrary values without affecting acceptance.

Now construct a graph $G_2$: a cycle of length $3i+3$ on nodes $v_0, v_1, \dots, v_{3i+2}$, where $v_k$ is adjacent to $v_{k+1}$ for every $0\le k<3i+2$, and $v_{3i+2}$ is adjacent to $v_0$.
We assign proof labels $L'(v_k) = l_{(k \bmod (i+1)) + 1}$ for all $k$.
Since $3(i+1)$ is a multiple of $i+1$, this proof label assignment fits exactly around the cycle.
We now show that all nodes in $G_2$ accept under $\mathcal{A}$.
Consider any node $v_k$ in $G_2$.
Since $3(i+1) > 2i+1$, the view of $v_k$ within distance $i$ does not wrap around the entire cycle, and thus consists of $2i+1$ consecutive degree-$2$ nodes forming a simple path---identical in structure to the view of any node in $S$.
The $2i+1$ labels in $v_k$'s view are a contiguous window from the periodic label pattern of period $i+1$.
Since each node sees $2i+1$ consecutive nodes while the sequence $(l_1, l_2, \ldots, l_{i+1})$ has length only $i+1$, every node $v_k$ in $G_2$ sees at least one complete occurrence of this sequence within its view.
That is, there exists an index $t_0$ with $0\le t_0\le i$ such that the proof labels appearing at positions $t_0, t_0 + 1, \ldots, t_0 + i$ within the view of $v_k$ are exactly $(l_1, l_2, \ldots, l_{i+1})$.
Now set $m = i - t_0 + 1$ so that $1 \leq m \leq i+1$). 
Consider the node $u_{i+m}$ in $G_1$. Its view also consists of $2i+1$ degree-$2$ nodes forming a path, and moreover the sequence $(l_1, l_2, \ldots, l_{i+1})$ appears in exactly the same positions $t_0, t_0 + 1, \ldots, t_0 + i$ within this view.
Therefore, the views of $u_{i+m}$ and $v_k$ have the same graph structure and agree on these $i+1$ proof labels. The remaining at most $i$ proof labels may differ. Since $\mathcal{A}$ tolerates up to $i$ adversarial modifications and accepts at $u_{i+m}$ in $G_1$, it must also accept at $v_k$ in $G_2$.

As $v_k$ was arbitrary, every node in $G_2$ accepts according to $\mathcal{A}$. Since $G_2$ contains a cycle, this contradicts the soundness of $\mathcal{A}$ for cycle-freeness.  
\end{proof}

\begin{proof}[$\blacktriangleright$ Impossibility Proof of Theorem \ref{thm:impos_cycle} for Bipartiteness:]

The impossibility proof for bipartiteness follows the same structure as the cycle-freeness proof, with one key modification.
Since bipartiteness requires the no-instance to be a \emph{non-bipartite} graph, we use an odd cycle as $G_2$.
However, the periodic proof label assignment used in the cycle-freeness proof requires the cycle length to be a multiple of $i+1$, which cannot be odd when $i$ is odd.
To handle this, we do the following: we take $G_1 = (u_0, u_1, \dots, u_{4i+4})$ to be a path of $4i+5$ nodes and $G_2 = (v_0, v_1, \dots, v_{2i+2})$ to be an odd cycle of length $2i+3$.
We assign proof labels of the nodes of $G_2$ by setting $L'(v_j) = L_O(u_{i+1+j})$ for all $j = 0, 1, \dots 2i+2$.

For any node $v_j$ in $G_2$, its view consists of exactly $2i+1$ consecutive degree-$2$ nodes forming a path, since the cycle length is $2i+3>2i+1$.
The corresponding node $u_{i+1+j}$ in $G_1$ also has a view of the same graph structure, because it lies at a distance of at least $i+1$ from both endpoints of the path.
Moreover, the proof labels in these two views agree at every corresponding position.
Hence, the view of $v_j$ in $G_2$ is identical to the view of $u_{i+1+j}$ in $G_1$.

Therefore, every node $v_j$ accepts according to $\mathcal A$.
Since $G_2$ is an odd cycle, it is not bipartite, which contradicts soundness.
\end{proof}

\section{CONGEST Implementation and Potential Future Directions}
\label{sec:congest-model}
One motivation for using proofs of size as small as possible is to adopt them in the CONGEST model. 
Unlike LOCAL, which places no constraint on per-edge bandwidth like CONGEST, a direct simulation of a LOCAL algorithm with view distance $k$ in CONGEST would require gathering the entire $k$-hop neighborhood, which may incur a prohibitive message size. 
We initiate a study of efficient CONGEST implementation for LCPs by considering the LCP for cycle existence with $2$ distinct proof labels ($0$ and $1$) and view distance $3$, as in Sec. \ref{sec:2label-3viewdistance-algorithm}.

We show that the verification algorithm $\mathcal{A}_0$\textsc{-2labels} (Algorithm~\ref{alg:A_0-algorithm-2labels}) can be implemented in the CONGEST model using only $3$ synchronous communication rounds, where each message consists of $O(1)$ bits (specifically, at most $4$ bits per message).

\vspace{1mm}
\noindent $\blacktriangleright$ \textbf{Highlevel Idea:} The algorithm proceeds in $3$ rounds. In round $1$, every node broadcasts its proof label (a single bit) to all its neighbors, so that each node learns the proof labels of its $1$-hop neighborhood.
In round~$2$, every node compresses the information about its $1$-hop neighborhood into one of $9$ predefined \emph{type messages} (requiring at most $4$ bits) and broadcasts it.
After receiving these messages, each node $v$ has enough information to simulate the \textsc{Discover-Parent}$(v)$ subroutine (Algorithm~\ref{alg:parent}), which helps $v$ decide whether $v$ is a core node, a tree node (and if so, identifies its parent), or whether $v$ should reject.
In round~$3$, each node broadcasts its decision so that neighboring nodes can verify mutual consistency, mirroring the consistency check in Algorithm~\ref{alg:A_0-algorithm-2labels}.

\vspace{1mm}
\noindent $\blacktriangleright$ \textbf{Type Messages.}
After round~$1$, each node $v$ knows the proof labels of all its neighbors. 
Based on this information, $v$ computes and broadcasts in round~$2$ a \emph{type message} $\tau(v)$, which is one of the following $9$ types and $\lceil \log_2 9 \rceil = 4$ bits suffice to encode $\tau(v)$.

\begin{quote}
\begin{enumerate}
    \item[\textbf{T1.}] $\deg(v) = 1$ (leaf node)
    
    \item[\textbf{T2.}] $\deg(v) = 2$ and both neighbors have the same proof label
    
    \item[\textbf{T3.}] $\deg(v) = 2$ and the two neighbors have different proof labels
    
    \item[\textbf{T4.}] $\deg(v) \geq 2$ and every neighbor $u \in N_1(v)$ has $L(u) = 0$
    
    \item[\textbf{T5.}] $\deg(v) \geq 2$ and every neighbor $u \in N_1(v)$ has $L(u) = 1$
    
    \item[\textbf{T6.}] $\deg(v) \geq 2$ and every neighbor $u \in N_1(v)$ except exactly one has $L(u) = 0$
    
    \item[\textbf{T7.}] $\deg(v) \geq 2$ and every neighbor $u \in N_1(v)$ except exactly one has $L(u) = 1$
    
    \item[\textbf{T8.}] $\deg(v) \geq 2$, $L(v) = 0$, and at least two neighbors $u, w$ of $v$ satisfy $L(u) = L(w) = 0$ 
    
    \item[\textbf{T9.}] None of the above
\end{enumerate}

\end{quote}

The rest of the algorithm is similar to $\mathcal{A}_0$\textsc{-2labels} (Algorithm~\ref{alg:A_0-algorithm-2labels}).

We present the pseudocode of the CONGEST variant (Algorithm~\ref{alg:congest-cycle-2labels}) and then prove its correctness, most of which follow from the analysis presented for the LOCAL variant of the problem.
In Algorithm~\ref{alg:congest-cycle-2labels}, 
during round~$3$, each node $v$ broadcasts a \emph{decision message} $\delta(v)$ encoding one of: ``I am a core node,'' ``My parent is $u$'' (where $u \in N_1(v)$ is identified by its port number), or ``I reject.''
Using the received decision messages, each node performs the same consistency check as in Lines \ref{2label-cycle-existence-line-start} - \ref{2label-cycle-existence-line-end} of Algorithm~\ref{alg:A_0-algorithm-2labels}.

\begin{algorithm2e}[!ht]
\DontPrintSemicolon
\tcc{\textbf{Round 1:} Share proof labels}
Broadcast $L(v)$ to all neighbors\;
Receive $L(u)$ from every neighbor $u \in N_1(v)$\;
\BlankLine
\tcc{\textbf{Round 2:} Share neighborhood type}
Compute type message $\tau(v)$ based on $L(v)$, $\{L(u) : u \in N_1(v)\}$, and $\deg(v)$\;
Broadcast $\tau(v)$ to all neighbors\;
Receive $\tau(u)$ from every neighbor $u \in N_1(v)$\;
\BlankLine
\tcc{Local computation: simulate \textsc{Discover-Parent}$(v)$}
Using $L(v)$, $\{L(u), \tau(u) : u \in N_1(v)\}$, determine:\;
\eIf{$v$ is a leaf (exactly one neighbor)}
    {$role(v) \leftarrow \texttt{tree}$; $parent(v) \leftarrow$ the unique neighbor\;}
    {
    \eIf{exactly one neighbor $p$ has $\deg(p) > 1$ (deduced from $\tau$)}
        {$role(v) \leftarrow \texttt{tree}$; $parent(v) \leftarrow p$\;}
        {
        \eIf{$v$ finds a valid base string $S_b$ consistent with all $v$-centered routes}
            {$role(v) \leftarrow \texttt{tree}$; $parent(v) \leftarrow$ node corresponding to the $2$nd bit of $S_b$\;}
            {
            \eIf{$L(v) = 0$ and at least two neighbors $u,w$ have $L(u) = L(w) = 0$}
                {$role(v) \leftarrow \texttt{core}$\;}
                {$role(v) \leftarrow \texttt{reject}$\;}
            }
        }
    }
\BlankLine
\tcc{\textbf{Round 3:} Consistency check}
\lIf{$role(v) = \texttt{reject}$}{\Return{\reject{}}}
Broadcast $\delta(v) = role(v)$ (and $parent(v)$ if tree node) to all neighbors\;
Receive $\delta(u)$ from every neighbor $u \in N_1(v)$\;
\BlankLine
\If{$role(v) = \texttt{core}$}
    {
        \lIf{at least $2$ neighbors $u$ have $\delta(u) = \texttt{core}$}{\Return \accept{}}
        \lElse{\Return \reject{}}
    }
\If{$role(v) = \texttt{tree}$ with $parent(v) = p$}
    {
        \lIf{$\delta(p) = \texttt{core}$ \textbf{or} ($\delta(p) = \texttt{tree}$ and $parent(p) \neq v$)}{\Return \accept{}}
        \lElse{\Return \reject{}}
    }
\caption{\textsc{CONGEST-Cycle-2Labels} (executed at node $v$)}
\label{alg:congest-cycle-2labels}
\end{algorithm2e}

\subsection*{Analysis of Algorithm \ref{alg:congest-cycle-2labels}.}
The key claim is that the type messages, together with the proof labels from round~$1$, carry enough information for every node to simulate the \textsc{Discover-Parent} subroutine and the subsequent consistency check of the LOCAL algorithm.
Once this is established, correctness follows from the existing analysis of the LOCAL variant.
Hence, we have the following theorem. 

\begin{theorem}
\label{thm:congest-cycle-existence}
    The proof $L_O^*$, consisting of $2$ distinct proof labels, together with Algorithm~\ref{alg:congest-cycle-2labels}, establishes an LCP for verifying cycle existence in $3$ synchronous CONGEST rounds, where each message consists of at most $4$ bits.
\end{theorem}

\begin{proof}[Proof of Theorem \ref{thm:congest-cycle-existence}]
We first prove the following key lemma.

\begin{lemma}[Simulation Correctness]
\label{lem:congest-simulation}
    After rounds~$1$ and~$2$ of Algorithm~\ref{alg:congest-cycle-2labels}, every node $v$ can determine the same output as the \textsc{Discover-Parent}$(v)$ subroutine (Algorithm~\ref{alg:parent}) using only the information $L(v)$, $\{L(u) : u \in N_1(v)\}$, and $\{\tau(u) : u \in N_1(v)\}$.
    Moreover, round~$3$ enables $v$ to perform the same consistency check as in Algorithm~\ref{alg:A_0-algorithm-2labels}.
\end{lemma}

\begin{proof}
    The proof of this lemma is constructive. 
    Recall that the \textsc{Discover-Parent}$(v)$ subroutine (Algorithm~\ref{alg:parent}) requires $v$ to handle four cases: (i)~$v$ is a leaf, (ii)~$v$ has exactly one non-leaf neighbor, (iii)~$v$ identifies a valid base string along a $v$-centered route, or (iv)~$v$ is a core node.
    After rounds~$1$ and~$2$, node $v$ knows its own label $L(v)$, the label $L(u)$ and the type message $\tau(u)$ for every neighbor $u \in N_1(v)$.
    We show case by case that this information suffices.

    Cases~(i) and~(iv) are immediate: for~(i), $v$ knows its own degree; for~(iv), $v$ knows $L(v)$ and $\{L(u) : u \in N_1(v)\}$ from round~$1$, which suffices to check whether $L(v)=0$ and at least two neighbors have label~$0$.
    Case~(ii) requires $v$ to detect which neighbors are leaves, which is immediate since $\tau(u) = \text{T1}$ if and only if $u$ is a leaf.

    It remains to handle case~(iii), in which $v$ must determine whether it lies on a $v$-centered route of $5$ nodes whose proof labels form a valid base string $S_b \sqsubseteq S_\infty$, and if so, identify its parent.
    For any $v$-centered route $w_1, u, v, u', w_2$, the length-$5$ string is $L(w_1) L(u) L(v) L(u') L(w_2)$.
    Node $v$ already knows $L(u)$, $L(v)$, and $L(u')$; hence it only needs to determine $L(w_1)$ and $L(w_2)$, i.e., the proof label of one $2$-hop neighbor through each of $u$ and $u'$.
    We now show that for each neighbor $u$ of $v$, the type message $\tau(u)$ together with $L(v)$ determines the proof label of the $2$-hop endpoint of any route through $u$, as follows.

    \begin{itemize}
        \item If $\tau(u) =$ T1: $u$ is a leaf, so there is no $2$-hop  neighbor through $u$. No $v$-centered route of $5$ nodes passes through $u$. 

        \item If $\tau(u) =$ T2: $u$ has exactly two neighbors with the same proof label. Since $v$ is one of them, the other neighbor $w$ of $u$ has the proof label $L(v)$. Hence, $L(w) = L(v)$.

        \item If $\tau(u) =$ T3: $u$ has exactly two neighbors with different labels. Since $v$ is one of them with label $L(v)$, the other neighbor $w$ has label $1-L(v)$.

        \item If $\tau(u) =$ T4 \text{ or } T5: All neighbors of $u$ share the same label ($0$ for T4, $1$ for T5). Every $2$-hop neighbor of $v$ through $u$ has that label.

        \item If $\tau(u) =$ T6: All neighbors of $u$ except exactly one have label~$0$, and the exception has label~$1$. If $L(v) = 1$, then $v$ is the unique exception, so every other neighbor of $u$ has label~$0$.
        If $L(v) = 0$, then the exception lies among $u$'s other neighbors: exactly one of them has label~$1$, while the rest have label~$0$.

        \item If $\tau(u) =$ T7: Symmetric to T6. All neighbors of $u$ except exactly one have label~$1$.
        If $L(v) = 0$, then $v$ is the exception, and all other neighbors of $u$ have label~$1$.
        If $L(v) = 1$, then the exception lies behind $u$.

        \item If $\tau(u) =$ T8: $u$ has identified itself as a potential core node (label~$0$ with at least two neighbors of label~$0$). This directly informs $v$ that $u$ may be a core node.

        \item If $\tau(u) =$ T9: $u$'s neighborhood does not match any of T1--T8. In particular, $u$ itself will reject, so no valid base string passes through $u$.
    \end{itemize}

    For types T1-T5, the label of every $2$-hop neighbor through $u$ is fully determined, so $v$ can construct the complete length-$5$ string for any $v$-centered route through $u$.
    For types T6 and T7, when $v$ is the exception (e.g., $\tau(u) = \text{T6}$ and $L(v) = 1$), all other neighbors of $u$ share the same label, so the $2$-hop endpoint label through $u$ is determined.
    When the exception lies behind $u$ (e.g., $\tau(u) = \text{T6}$ and $L(v) = 0$), exactly one of $u$'s other neighbors has a different label.
    In this case, $v$ knows that there are two possible labels ($0$ and $1$) for the $2$-hop endpoint behind $u$, depending on whether the route passes through the exceptional neighbor of $u$ or not. 
    However, this is sufficient for the \textsc{Discover-Parent} subroutine: for each candidate base string $S_b$, $v$ can check whether the $2$-hop proof label required by $S_b$ behind $u$ is consistent with the type message $\tau(u)$.
    Specifically, $v$ enumerates the candidate base strings for each $v$-centered route and verifies whether all other length-$5$ strings from $v$-centered routes belong to the tree string set $t(S_b)$ (Table~\ref{table:refined-membership}).
    Since no base string is a palindrome and there is no pair $S_{bi}, S_{bj}$ with both $S_{bi} \in t(S_{bj})$ and $S_{bj} \in t(S_{bi})$, the direction is uniquely determined, and $v$ identifies its parent correctly.

    Finally, the consistency check in round~$3$ requires each node $v$ to compare its own role (core/tree with parent/reject) against the decisions of its neighbors.
    This is identical to Algorithm~\ref{alg:A_0-algorithm-2labels}: a core node must have at least two core-node neighbors, and a tree node $v$ with parent $p$ must verify that $p$ does not claim $v$ as its parent (or that $p$ is a core node).
    Since round~$3$ provides $v$ with exactly the decision messages $\{\delta(u) : u \in N_1(v)\}$, this check can be performed.

    Therefore, the output of every node $v$ in Algorithm~\ref{alg:congest-cycle-2labels} is identical to the output of $\mathcal{A}_0$\textsc{-2labels} (Algorithm~\ref{alg:A_0-algorithm-2labels}).
\end{proof}

Rest of the proof of Theorem \ref{thm:congest-cycle-existence} follows from analysis of the LOCAL variant. 
By Lemma~\ref{lem:congest-simulation}, Algorithm~\ref{alg:congest-cycle-2labels} produces the same accept/reject decisions as the LOCAL verification algorithm $\mathcal{A}_0$\textsc{-2labels} (Algorithm~\ref{alg:A_0-algorithm-2labels}) at every node.
Lemma~\ref{lem:2labels-3view-graph-with-cycles} and Lemma~\ref{lem:2labels-3view-graph-without-cycles} establish the completeness and soundness properties. 
\end{proof}

This unfolds some interesting open problems, and we conclude with them. 

\vspace{1mm}

\noindent $\blacktriangleright$ \textbf{Open Problem 1:} Can we implement LCPs with errors efficiently in the CONGEST model, where nodes need access to a larger neighborhood depending on the number of errors?

\vspace{1mm}

\noindent $\blacktriangleright$ \textbf{Open Problem 2:}  In the LOCAL model, which of the problems mentioned in \cite{goos2016locally} admit LCPs even with errors, and moreover, can we modify \framework{} to achieve optimal view distance?

\vspace{1mm}

\noindent $\blacktriangleright$ \textbf{Open Problem 3:} Our $2$-label cycle-existence construction based on the repeated string $001101$ suggests a more general question: can similar pattern-based gadgets be used to encode local direction or other structural information for different verification problems?

\subsection*{Acknowledgment:} This work is supported by Polish National Science Centre project no.\ 2020/39/B/ST6/03288.

\newpage
\appendix

\bibliographystyle{plainurl}
\bibliography{bibliography}

\end{document}